\documentclass[journal]{IEEEtran}
\usepackage{amsmath}
\usepackage{bm}
\usepackage{subfigure}
\usepackage{graphicx,epsfig,balance}
\usepackage{cite}
\usepackage{color}
\usepackage{multirow}
\usepackage{amsfonts,amssymb}
\usepackage{amsmath, amsthm}
\usepackage{booktabs}
\usepackage{threeparttable}
\usepackage[ruled]{algorithm2e}
\newtheorem{theorem}{Theorem}

\newtheorem{remark}{Remark}

\usepackage{diagbox}
\usepackage{soul}

\soulregister\cite7 % 针对\cite命令
\soulregister\citep7 % 针对\citep命令
\soulregister\citet7 % 针对\citet命令
\soulregister\ref7 % 针对\ref命令
\soulregister\pageref7 % 针对\pageref命令
\soulregister\item7 % 针对\pageref命令
\soulregister\itemize7 % 针对\pageref命令
\soulregister\align7 % 针对\pageref命令
\soulregister\split7 % 针对\pageref命令
\soulregister\figure7 % 针对\pageref命令
\soulregister\caption7 % 针对\pageref命令
\usepackage{tcolorbox}
\usepackage{amsmath,amssymb}
\usepackage{mathrsfs}
\usepackage[letterpaper, top=1.92cm, bottom=3.2cm, left=1.92cm, right=1.92cm]{geometry}
\usepackage[ruled]{algorithm2e}
\usepackage{algpseudocode}
\usepackage{xcolor}
\usepackage{color}
\definecolor{commentcolor}{RGB}{169,169,169} % Define the light gray color
\algrenewcommand\algorithmiccomment[1]{\textcolor{commentcolor}{#1}}
\usepackage{setspace}
\usepackage{makecell}
\usepackage{url} 
\usepackage{hyperref}
\usepackage{orcidlink}
\usepackage{tikz}

\usepackage{soul}
\newtheorem{assumption}{Assumption}
\newtheorem{definition}{Definition}
\begin{document}
	\title{Robust Deep Joint Source-Channel Coding Enabled Distributed Image Transmission with Imperfect Channel State Information
	}
%	\author{Biao Dong,
%		Bin Cao, 
%		Jingjing Luo, 
%		and Yu Wang, \IEEEmembership{Member, IEEE}
%		%  	and Rongxing Lu, \IEEEmembership{Fellow, IEEE}
	% 	\thanks{}
	% 	\IEEEcompsocitemizethanks{
	% 		\IEEEcompsocthanksitem (\textit{Corresponding author}: \textit{Bin Cao})
	% 		\IEEEcompsocthanksitem Biao Dong, Bin Cao and Jingjing Luo are with the School of Electronic and Information Engineering, Harbin Institute of Technology (Shenzhen), Shenzhen 518055, China (e-mail: 23b952012@stu.hit.edu.cn; caobin@hit.edu.cn; luojingjing@ hit.edu.cn).
	% 		\IEEEcompsocthanksitem Yu Wang is with the College of Telecommunications and Information Engineering, Nanjing University of Posts and Telecommunications, Nanjing 210003, China (e-mail: yuwang@njupt.edu.cn).
	% 		%	\IEEEcompsocthanksitem Rongxing Lu is with the Faculty of Computer Science, University of New Brunswick, Fredericton, NB E3B 5A3, Canada (e-mail: rlu1@unb.ca).
	% }}
		\author{
        Biao Dong,
	Bin Cao, \IEEEmembership{Member, IEEE},\\
        Guan Gui, \IEEEmembership{Fellow, IEEE},
        and Qinyu Zhang, \IEEEmembership{Senior Member, IEEE}
	%  	and Rongxing Lu, \IEEEmembership{Fellow, IEEE}
		\thanks{}
		\IEEEcompsocitemizethanks{This work was supported in part by the Shenzhen Science and Technology Program under Grant KJZD20240903095402004.
                Parts of this paper were presented at the IEEE Global Communications Conference (GLOBECOM), Cape Town, South Africa, Dec. 2024 \cite{Dong2024GC}. The open source code of this work is available at: {\href{https://dongbiao26.github.io/rdjscc/}{https://dongbiao26.github.io/rdjscc/}}.\IEEEcompsocthanksitem 
			\IEEEcompsocthanksitem Biao Dong, Bin Cao and Qinyu Zhang are with the School of Electronic and Information Engineering, Harbin Institute of Technology (Shenzhen), Shenzhen 518055, China (e-mail: 23b952012@stu.hit.edu.cn; caobin@hit.edu.cn; zqy@hit.edu.cn).
			\IEEEcompsocthanksitem Guan Gui is with the College of Telecommunications and Information Engineering, Nanjing University of Posts and Telecommunications, Nanjing 210003, China (e-mail: guiguan@njupt.edu.cn).
			%	\IEEEcompsocthanksitem Rongxing Lu is with the Faculty of Computer Science, University of New Brunswick, Fredericton, NB E3B 5A3, Canada (e-mail: rlu1@unb.ca).
	}}

	\markboth{IEEE Transactions on Wireless Communications,~Vol.~XX, No.~XX, Month Year}{}
	\maketitle
\begin{abstract}
This work is concerned with robust distributed multi-view image transmission over a severe fading channel with imperfect channel state information (CSI), wherein the sources are slightly correlated. In contrast to point-to-point deep joint source-channel coding (DJSCC), the distributed setting introduces the key challenge of exploiting inter-source correlations without direct communication, especially under imperfect CSI. 
% Since the signals are further distorted at the decoder, traditional distributed deep joint source-channel coding (DJSCC) suffers considerable performance degradation. 
To tackle this problem, we leverage the complementarity and consistency characteristics among the distributed, yet correlated sources, and propose an robust distributed DJSCC, namely RDJSCC. In RDJSCC, we design a novel cross-view information extraction (CVIE) mechanism to capture more nuanced cross-view patterns and dependencies. In addition, a complementarity-consistency fusion (CCF) mechanism is utilized to fuse the complementarity and consistency from multi-view information in a symmetric and compact manner. Theoretical analysis and simulation results show that our proposed RDJSCC can effectively leverage the advantages of correlated sources even under severe fading conditions, leading to an improved reconstruction performance. 

\end{abstract}


\IEEEpeerreviewmaketitle
%It is necessary to design flexible multi-view signal optimization to meet the requirements of consistency and complementarity in various scenarios.

\section{Introduction}
\subsection{Background}
%%%TODO:归纳逻辑，使用边学习，如何高效使用，如何在有噪的情况下使用
% 
% The 6th-Generation mobile communications system mobile networks are envisioned to provide larger coverage areas and higher quality services \cite{You2021China}. 
Various wireless applications, such as autonomous driving, remote healthcare, virtual reality, etc,  increase the pressure on wireless sensor networks. Given the limitations of computing resources, how to efficiently utilize the correlation between distributed sources at the decoder has attracted more and more attention. It is believed that distributed source coding (DSC) is a promising approach to achieving high quality multimodal communications, since DSC leverages the correlation of distributed sources and enables low-complexity encoding by shifting a significant amount of computation to the decoder \cite{Gündüz2022JSAC,Wolf1973TIT,WynerZiv1976,Heegard1985}. 

Traditional DSC primarily focuses on the information-theoretic perspective, such as Slepian-Wolf coding \cite{Wolf1973TIT}, Wyner-Ziv coding \cite{WynerZiv1976}, Berger-Tung coding \cite{Heegard1985}, etc. Although the theoretical framework of DSC paves a solid foundation to handle the compression problems of distributed sources, practical DSC systems have not been widely used due to the challenge in capturing complex correlation among sources in severe fading environments. 

% It is also expected to are envisioned to be intelligent transmission with extremely low end-to-end latency, multi-node efficient cooperative communication and  multi-sensor data fusion .
% artificial intelligence (AI)-enabled human-machine interaction . 
% In addition, even in severe fading environment (e.g., autonomous driving scenario in Fig. \ref{intro}).  
% Distributed source coding (DSC) 
% Semantic communications emerge as a promising solution for these visions, which relies on deep learning models to achieve semantic-level signal processes\cite{Bourtsoulatze2019TCCN,Zhang2023TWC, Xu2022TCSVT,Dai2022JSAC}. As the cornerstone of semantic communication technology, recent research on DL-based joint source-channel coding (DJSCC) indicates that DJSCC can dynamically allocating bandwidth to source or channel coding and thus present a graceful performance degradation in fading environment \cite{Bourtsoulatze2019TCCN}. Further, \cite{Zhang2023TWC, Xu2022TCSVT,Dai2022JSAC} extended the DJSCC to a dynamic channel environments and allocate bandwidth based on image content, achieving a higher end-to-end transmission performance. 

\begin{figure}[htbp]
	\centering
	\includegraphics[height=5 cm, width=8cm] {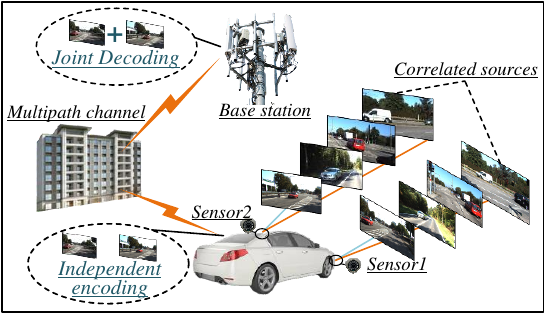}
	\caption{An autonomous driving example for illustrating DSC. Two sensors capture different views of the same obstacle, and these two views form correlated sources. The two views are independently encoded, transmitted through noisy channels, and finally jointly decoded at the base station.}
	\label{intro4}
\end{figure} 

\subsection{Related Works}
Recently, deep learning (DL) has made great success in many fields, such as channel estimation \cite{Hu2021}, intelligent signal processing \cite{Dong2022}, and multi-access communication \cite{Liu2022}. The popularity of DL can be summarized as follows: unlike human-crafted models, DL leverages deep neural networks (DNNs) to extract complex features automatically. Secondly, complex optimization problems can be tackled in an end-to-end learning way based on various DNN layers. Hence, more and more works attempted to apply DL to DSC for improving compression performance. 
% \subsubsection{DL-based DSC}
The pioneering DL-based DSC framework was proposed in \cite{Liu2019ICCV}, which utilized the mutual information (MI) between distributed images for efficient compression. \cite{Mital2022DCC} further improved this compression framework, by extracting common information rather than feeding correlated sources to the decoder directly. 
On this basis, \cite{Li2024NIPS} extended the existing DSC to task-aware scenarios, which can allocate bandwidth based on task attributes and thus achieve an elegant rate-distortion trade-off. However, most of these works focused on lossy source compression, while neglecting the impact of wireless channel imperfections. Fig. \ref{intro4} presents an autonomous driving scenario to illustrate DSC under a noisy channel. In this regard, how to design efficient DSC to guarantee reliable transmission of correlated sources is a concerning issue. To this end, a DL-based DSC scheme which can ensure reliable transmission in noisy channels is needed. 
% There has been a growing interest in utilizing the correlated source in an efficient way, interested readers can refer to \cite{OzyilkanErkip2024CISS}. 

An effective approach to exploring the aforementioned problem is to utilize DL-based joint source and channel coding (DJSCC) \cite{Bourtsoulatze2019TCCN}. DJSCC leverages DL model to directly map sources to channel inputs. As shown in Fig. \ref{intro5}(a), source $\mathbf{s}_1$ is mapped as channel inputs $\mathbf{x}_1$ by a DNN. Recent research on DJSCC indicates that DJSCC can dynamically allocate bandwidth to source or channel coding and thus present a graceful performance in fading environments \cite{Bourtsoulatze2019TCCN,Zhang2023TWC, Xu2022TCSVT,Dai2022JSAC,Wu2022WCL,Yang2022TCCN, Shao2023WCL}. The first DJSCC method for wireless image transmission was proposed in \cite{Bourtsoulatze2019TCCN}, which proved the superiority of DJSCC in low signal-to-noise ratio (SNR) environments. 
% Further, \cite{Zhang2023TWC, Xu2022TCSVT,Dai2022JSAC} extended the DJSCC to a dynamic channel environments and allocate bandwidth based on image content, achieving a higher end-to-end transmission performance. 
To further improve the rate-distortion (RD) performance, \cite{Dai2022JSAC} introduced a hyperprior as side information and integrated Swin Transformer as the backbone. Combining the advantages of digital modulation, \cite{Yang2022TCCN,Wu2022WCL, Shao2023WCL} attempted to deploy orthogonal frequency division multiplexing (OFDM) to DJSCC. To utilize a single DJSCC across different SNR conditions and compression ratios, adaptive strategies of SNRs or compression ratios were also extensively studied \cite{Zhang2023TWC, Xu2022TCSVT,Dai2022JSAC}. 
% In addition, task-oriented DJSCC also attracted great attentions \cite{Shao2022JSAC}, i.e., jointly optimizing DJSCC with downstream tasks in a task-oriented manner to achieve edge inference with low latency. 

In this paper, we devote to extending the DJSCC to distributed communication among multiple correlated sources over fading channels with imperfect channel state information (CSI). Unlike point-to-point DJSCC, which deals with a single source and transmission path, our setting aligns with DSC, where multiple spatially separated but statistically correlated sources independently encode their data without inter-source communication. The encoded signals are then transmitted over noisy channels and jointly decoded at a central decoder, which is referred to as distributed DJSCC.  
From an information-theoretic viewpoint, this constitutes a class of network capacity problem, rather than the single-user capacity scenario in point-to-point DJSCC. In such settings, effectively leveraging inter-source correlation becomes critical to approaching theoretical performance limits, as discussed in \cite{Gündüz2022JSAC, Wolf1973TIT, WynerZiv1976, Heegard1985}. We notice that two existing works have partially studied this issue \cite{Yilmaz2023arx, Wang2022ICASSP}. The authors in \cite{Yilmaz2023arx}, as shown in Fig. \ref{intro5}(a), proposed a low-latency image transmission method when one of correlated sources is losslessly accessed at the receiver. A novel neural network architecture incorporating the lossless correlated information at multiple stages was designed at the decoder. Nevertheless, in practical communication systems, correlated sources are often not losslessly accessed, as shown in Fig. \ref{intro5}(b). The channel quality of each source may vary significantly over time. Mismatched or low-correlated sources can lead to a degradation of transmission performance. In \cite{Wang2022ICASSP}, the authors considered the lossy access situation, and utilized a cross attention mechanism (CAM)-based DJSCC to capture the complex correlation among distributed sources. Simulation results demonstrated that CAM can achieve an improvement of reconstruction quality. However, CAM-based DJSCC focused on the additive white Gaussian noise (AWGN) channel and Rayleigh channel with perfect CSI, neglecting the impacts of severe fading with imperfect CSI on the correlation of sources. For example, in Fig. \ref{intro4}, the rapid movement of vehicles leads to serve channel fading, and thus perfect CSI acquisition is challenging. In this case, the correlation of sources undergoes drastic fluctuations. In addition, CAM-based DJSCC does not consider the trade-off between complementarity and consistency of lossy correlated sources.

% Complementarity-consistency is a concept widely used in multi-view learning but unexplored in distributed image transmission \cite{Li2022TKD}. 
% In short, best balance between complementarity and consistency can lead to the optimal reconstruction performance. As a conclusion, we can see that the lossy access of correlated source is still not effectively addressed.

% The majority of existing works about DJSCC are point-to-point communication, which do not consider distributed communication among multiple nodes. For example, in distributed sensors network of autonomous driving (Fig. \ref{intro}), it is challenge that effectively encoding correlated source distributed across different locations without global communication during the encoding process \cite{OzyilkanErkip2024CISS}. Although the theoretical framework of distributed source coding (DSC) provides a promising approach to handle data compression problems \cite{WynerZiv1976}, practical DSC system has not widespread use due to the challenge of designing efficient DSC to capture complex correlation between sources. 

\subsection{Contributions}
Motivated by the aforementioned perspectives, a robust deep joint source-channel coding (RDJSCC) enabled distributed image transmission scheme is proposed in this work. We focus on a more realistic (or general) scenario where correlated sources using OFDM modulation are lossy access over multi-path channel without perfect CSI. Our goal is to maximize the advantages of distributed source coding in noisy environments. Specifically, the contributions of this paper can be summarized as follows.
\begin{itemize}	
	\item  Guided by theoretic analysis, we design a flexible multi-view transmission framework to meet the requirements of consistency and complementarity. Specifically, a new RDJSCC enabled distributed image transmission scheme is proposed. Unlike CAM-based DJSCC, RDJSCC explores the trade-off between complementarity and consistency. To the best of our knowledge, this is the first work exploiting the complementarity and consistency to maximize the advantages of distributed DJSCC without perfect CSI.
    \item We respectively develop a novel cross-view information extraction (CVIE) mechanism and complementarity-consistency fusion (CCF) mechanism in RDJSCC to upgrade distributed DJSCC. Specifically, CVIE can learn cross-view information efficiently based on the shift mechanism. CCF could fuse the complementarity and consistency from multi-view information in a symmetric and compact manner based on a dynamic
    weight assignment policy.   
	% \item We propose a new definition to extend the trade-off between consistency and complementarity to reconstruction level. In addition, based on variational analysis, we find that RDJSCC can be summarized as maximum likelihood estimation problem. Information-theoretic analysis is provided for why need design flexible multi-view signal optimization to meet the requirements of consistency and complementarity in distributed image transmission scenario. 
    \item Theoretical analysis and numerical experiments are conducted. Compared with CAM-based DJSCC, RDJSCC has better performance in terms of various indicators such as peak signal-to-noise ratio (PSNR), multi-scale structural
similarity index (MS-SSIM), and learned perceptual image patch similarity (LPIPS). We also verify the trade-off between peak-to-average power ratio (PAPR) and transmission performance in terms of PSNR. 
\end{itemize}	

\subsection{Organization}
The rest of this paper is arranged as follows. In Section \ref{sec:Problem}, we briefly introduce the considered system model. Section \ref{sec:Analysis} gives the theoretic analysis of RDJSCC, including the consistency and complementarity analysis at the reconstruction level. Guided by the theoretic analysis, novel CVIE and CCF are proposed in Section \ref{sec:Decoding}. Numerical results and discussions are given in Section \ref{sec:Experiments}, followed by conclusions in Section \ref{sec:Conclusion}. The summary of major notations is shown in Table \ref{notations}.

\begin{table}[htbp]
	\caption{Summary of Major Notations.}
	\centering
	\small
	\begin{tabular}{|c|l|l}
		\cline{1-2}
		\textbf{Notation}               & \textbf{Definition}    &  \\ \cline{1-2}
		$\mathbf{s}_{1},\mathbf{s}_{2}$ & Two correlated sources from two views &  \\ 
		$\mathbf{x}_{1},\mathbf{x}_{2}$ & Compressed representations of $\mathbf{s}_{1},\mathbf{s}_{2}$ &  \\ 
		${(\mathbf{z}_{1}, \mathbf{z}_{2})}$  & Corrupted version of $(\mathbf{x}_{1}, \mathbf{x}_{2})$ by channel &  \\ 
		$f(\cdot;\boldsymbol \phi)$&  DL-based encoder function parameterized with $\boldsymbol \phi$ &  \\ 
		$g(\cdot;\boldsymbol \theta)$&  DL-based decoder function parameterized with $\boldsymbol \theta$ &  \\ 
		$h$&   Impulse response of the multipath channel &  \\ 
		$\mathcal{K}$&    Dynamic weight &  \\ 		
		$d(\cdot)$&  Mean square error (MSE) &  \\ 
		$\rho$&  Clipping ratio  &  \\
		$R$&  Compression ratio  &  \\ \cline{1-2}
	\end{tabular}
 \label{notations}
\end{table}

\section{Problem Formulation} \label{sec:Problem} 
First, we introduce the system model. Then, we extend distributed DJSCC to OFDM-based communication systems.
\subsection{System Model}
We consider the following distributed images uplink transmission with a sensor $\mathbf{s}_{1}\in\mathbb{R}^M$ and its correlated version $\mathbf{s}_{2}\in\mathbb{R}^M$, with a joint distribution $p(\mathbf{s}_{1},\mathbf{s}_{2})$ capturing two views of the same object. Two sensors independently transmit their compressed representations $(\mathbf{x}_{1}, \mathbf{x}_{2})$ to a central decoder over a fading channel for joint decoding, as shown in Fig. \ref{intro5}(b). 
% To explore the impacts of severe fading on correlated sources, this paper focuses on the orthogonal multiple access (OMA).
The compressed pair $(\mathbf{x}_{1}, \mathbf{x}_{2})$ is encoded from the originally correlated sources $(\mathbf{s}_{1}, \mathbf{s}_{2})$. We define $(\mathbf{z}_{1}, \mathbf{z}_{2}$ as the corrupted version of $(\mathbf{x}_{1}, \mathbf{x}_{2})$ by channel. The Wyner-Ziv theorem points out that independent encoding and joint decoding of correlated sources can theoretically achieve the same compression ratio as a joint encoding-decoding scheme under lossy compression \cite{WynerZiv1976}.  

\begin{figure}[htbp]
	\centering
	\includegraphics[height=5 cm, width=7cm] {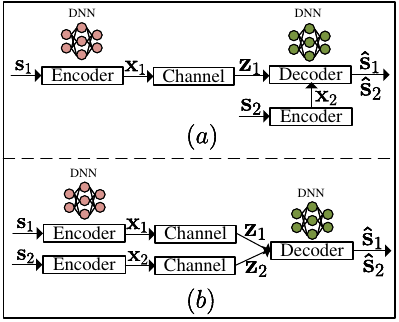}
	\caption{(a) The system model used in \cite{Yilmaz2023arx} where one of correlated sources is losslessly accessed. (b) Our considered system model where both correlated sources are lossy accessed.}
	\label{intro5}
\end{figure} 

\begin{remark}
	\label{remark1}
    Unlike the setting in \cite{Yilmaz2023arx}, where the decoder can losslessly access one source, as shown in Fig. \ref{intro5}(a), we consider the decoder access both correlated sources undergoing channel fading, as shown in Fig. \ref{intro5}(b), which reflects real-world scenarios.
\end{remark}

The correlated sources adopt DL-based encoder-decoder pair for codec training. Let $f(\cdot;\boldsymbol \phi)$ denote DL-based encoder function parameterized as ${\boldsymbol{\phi}}$. We further define $\mathbf{x}=f(\mathbf{s};\boldsymbol \phi)$, $\mathbf{x} \in \mathbb{C}^{M^{'}}$. The compression ratio can be obtained as $R \triangleq {M^{'}}/M$. The DNN-based codec architecture is given in Appendix \ref{sec:Architecture}. Over a block fading channel, the received signal can be expressed as
% \begin{equation}
% 	\label{channel}
% 	\mathbf{y}=h(\mathbf{ x};\sigma_0^2,\ldots,\sigma_{L-1}^2,\sigma^2)=h*\mathbf{x}+w,
% \end{equation}

\begin{equation}
	\label{channel}
	{\mathbf{z}=\mathbf{h}*\mathbf{x}+\mathbf{w},}
\end{equation}
where $*$ denotes a linear convolution implemented with zero-padding to ensure the output $\mathbf{z} \in \mathbb{C}^{M'}$ has the same length as the input $\mathbf{x} \in \mathbb{C}^{M'}$; $\mathbf{w} \sim \mathcal{CN}(0, \sigma^2 \mathbf{I}_{M^{'}\times{M^{'}}})$ denotes AWGN; $\mathbf{h} \in \mathbb{C}^{L}$ is the impulse response vector of an $L$-path multipath channel; Each path's channel coefficient $h_l$ follows a complex Gaussian distribution with zero mean and variance $\sigma_l^2$, i.e., $h_l\sim\mathcal{CN}(0,\sigma_l^2)$ for $l=0,1,\ldots,L-1$. The variance $ \sigma_l^2$ follows the exponential decay, i.e.,  $\sigma_l^2 = \alpha_l e^{-\frac{l}{\gamma}}$, where $\gamma$ represents the delay and $\alpha_l$ is a normalization coefficient. The sum of variances equals to 1, i.e., $\sum_{l=0}^{L-1} \sigma_l^2 = 1$. 
%%%TODO:优化问题的形式表达
The decoder at the base station recover the transmitted images as $\hat{\mathbf{s}}_{1}, \hat{\mathbf{s}}_{2}=g({\mathbf{z}}_{1},\mathbf{z}_{2};\boldsymbol{\theta})$, where $g(\cdot;{\boldsymbol{\theta}})$ is the decoder function parameterized as ${\boldsymbol{\theta}}$, and $\hat{\mathbf{s}}_{1}, \hat{\mathbf{s}}_{2}\in\mathbb{R}^N$ are the recovered images.

\subsection{OFDM-based DJSCC} \label{sec:Modulation}
%Orthogonal Frequency Division Multiplexing (OFDM) is a commonly used multi-carrier modulation technique. 
%Inspired by \cite{Yang2022TCCN}, which extends DJSCC to OFDM-based system.
Next, we describe how to introduce OFDM to distributed DJSCC.
The detailed process is summarized in Algorithm~\ref{algorithm1}. Each encoded representation $\mathbf{x}$ is power normalized and then allocated with an OFDM packet. Each packet contains $N_s$ information symbols and $N_p$ pilot symbols. The pilot symbols $\mathbf{x}_p\in\mathbb{C}^{N_p \times N_c}$ are known to both the transmitter and receiver. Under the OFDM modulation setting, $\mathbf{x}$ represents the frequency domain symbol. First, $\mathbf{x}$ selects $N_c$ subcarriers, i.e.,  $\mathbf{x}$ is reshaped as $\mathbf{x}\in\mathbb{C}^{N_s \times N_c}$. Then, symbols on all subcarriers can be transformed into a single time-domain symbol $ \mathscr{X}\in\mathbb{C}^{N_s \times N_c} $ through inverse discrete Fourier transform (IDFT) denoted as, $ \mathscr{X}=\boldsymbol{F}_{N_c}^H \mathbf{x}$, where $\boldsymbol{F}_{N_c}$ is an ${N_c}$-dimensional DFT matrix and $\boldsymbol{F}_{N_c}^H$ is an IDFT matrix. Next, a cyclic prefix (CP) of length $L_{\mathrm {cp}}$ is added on the time-domain symbol $\mathscr{X}$ yielding $\mathscr{X}_{\mathrm {cp}}\in\mathbb{C}^{N_s \times (N_c+L_{\mathrm {cp}})}$. The pilot symbols $\mathbf{x}_p$ are also transformed using IDFT and appended with a CP, then concatenated with $\mathscr{X}_{\mathrm{cp}}$ to form the OFDM symbol $\mathscr{X}_{\mathrm{ofdm}} \in \mathbb{C}^{(N_s + N_p) \times (N_c + L_{\mathrm{cp}})}$.

After adding pilots, $\mathscr{X}_{\mathrm {ofdm}}$ is transmitted through the fading channel as Eq. (\ref{channel}). When the base station receives $\mathscr{\hat X}_{\mathrm {ofdm}}$ (the corrupted version of $\mathscr{X}_{\mathrm {ofdm}}$), the OFDM demodulation is performed, which involves removing the CP and obtaining the corrupted frequency-domain symbols $\mathbf{z}$ using DFT, as well as the corrupted pilot symbols $\mathbf{z}_p$.

\begin{algorithm}
	\caption{Signal Processing Procedure of OFDM-based distributed DJSCC.}%算法名字
	\label{algorithm1}
	\LinesNumbered %要求显示行号
	\KwIn{Training data $(\mathbf{s}_{1},\mathbf{s}_{2})$; The pilot symbols $\mathbf{x}_p\in\mathbb{C}^{N_p \times N_c}$;}
	\KwOut{The recovered images $\hat{\mathbf{s}}_{1}, \hat{\mathbf{s}}_{2}$}%输出\\
	\For{$(\mathbf{s}_{1},\mathbf{s}_{2}) \in \mathcal{D}_{\mathrm{train}}$}
	{		
        ${SNR}_1$, ${SNR}_2$ $\leftarrow$ Randomly generate uniform SNR of independent channels;\leavevmode \\
        \Comment{$\triangleright$ Encoding} \leavevmode \\
		$\mathbf{x}_{1} \in \mathbb{C}^{N_s \times N_c}\leftarrow f(\mathbf{s}_{1}, {SNR}_1;\boldsymbol \phi)$;\leavevmode \\
      $\mathbf{x}_{2} \in \mathbb{C}^{N_s \times N_c}\leftarrow f(\mathbf{s}_{2}, {SNR}_2;\boldsymbol \phi)$;
		
		\Comment{$\triangleright$ Modulation} \leavevmode \\
		$\mathscr{X}\in \mathbb{C}^{N_s \times N_c} \leftarrow \boldsymbol{F}_{N_c}(\mathbf{x});$
		
		$\mathscr{X}_{\mathrm {cp}}\in \mathbb{C}^{(N_s+L_{cp}) \times N_c} \leftarrow \mathrm{CP}(\mathscr{X}, CP);$
		
		$\mathscr{X}_{\mathrm {ofdm}}\in \mathbb{C}^{(N_s+L_{cp}) \times (N_c+N_p)} \leftarrow \mathrm{Pilot}(\mathscr{X}_{\mathrm {cp}}, \mathbf{x}_p);$
		
		%		$\mathscr{X}_{\mathrm {ofdm}}^{clip}\in \mathbb{C}^{(N_s+L_{cp}) \times (N_c+N_p)} \leftarrow \mathrm{Clip}(\mathscr{X}_{\mathrm {ofdm}}, \mathbf{x}_p);$
		$\mathscr{X}_{\mathrm {ofdm}}^{clip} \leftarrow \mathrm{Clip}(\mathscr{X}_{\mathrm {ofdm}}, \mathbf{x}_p);$
		
		\Comment{$\triangleright$Through the fading channel} \leavevmode \\
		$\mathscr{\hat X}_{\mathrm {ofdm}}^{clip}\leftarrow \mathbf{h}*\mathscr{X}_{\mathrm {ofdm}}^{clip}+\mathbf{w};$\leavevmode \\				
		$\mathbf{z}\leftarrow \mathscr{\hat X}_{\mathrm {ofdm}}^{clip};$ \Comment{$\triangleright$Demodulation}\leavevmode \\

        \Comment{$\triangleright$ Dcoding}  \leavevmode \\
		$\hat{\mathbf{s}}=g({\mathbf{z}_{1}},{\mathbf{z}_{2}}, {SNR}_1, {SNR}_2;\boldsymbol{\theta});$	
	}	
\end{algorithm}

\begin{remark}
	\label{remark2}
    It should be pointed out that the proposed method belongs to discrete-time analog transmission (DTAT), where baseband complex symbols are transmitted directly after OFDM modulation, instead of passband transmission of digital signals \cite{Shao2023WCL}. This approach combines the advantages inherent to OFDM (i.e., overcoming frequency-selective fading) with DTAT. However, there exists a trade-off between PAPR and performance in OFDM-based DJSCC system. 
\end{remark}

The PAPR problem is important in the uplink since the efficiency of power amplifier is critical due to the limited battery power in a mobile terminal \cite{Yong2010book}. \cite{Shao2023WCL} proved that the high PAPR of DJSCC could be tackled by incorporating clipping into the training process as,
\begin{equation}
	\mathscr{X}_{\mathrm {ofdm}}^{clip}= \begin{cases}\mathscr{X}_{\mathrm {ofdm}}, & \text { if }\left|\mathscr{X}_{\mathrm {ofdm}}\right| \leq \rho \mathscr{\bar X}_{\mathrm {ofdm}} \\ \rho \mathscr{\bar X}_{\mathrm {ofdm}}, & \text { if }\left|\mathscr{X}_{\mathrm {ofdm}}\right|>\rho \mathscr{\bar X}_{\mathrm {ofdm}},\end{cases}
\end{equation}
where $\mathscr{X}_{\mathrm {ofdm}}^{clip}$ represents the clipped signal, $\rho$ represents the clipping ratio and $\mathscr{\bar X}_{\mathrm {ofdm}}$ represents the average amplitude of $\mathscr{X}_{\mathrm {ofdm}}$. The clipping ratio $\rho$ should be selected properly to avoid destroying the orthogonality among subcarriers. We incorporate clipping into the training process of RDJSCC to strike the balance between PAPR and performance.

\section{Preliminary Theoretic Analysis of the Pproposed RDJSCC} \label{sec:Analysis}
In this section, we first define reconstruction-relevant information for exploring how much information is sufficient for the reconstruction of the correlated images. Then, we build
a variational model for distributed image transmission and provide an analysis of imperfect CSI. Finally, we give an information-theoretic analysis regarding the impact of noise on the correlated sources, which guides us in efficiently utilizing MI to achieve collaborative recovery in the practical distributed wireless sensors network.

\subsection{Complementarity and Consistency Analysis} \label{sec:Analysis1}
We begin by analyzing the problem of distributed image transmission and ask how much information is sufficient for the reconstruction of multi-view images. 

\begin{figure}[htbp]
	\centering
	\includegraphics[height=2.5cm, width=8.5cm] {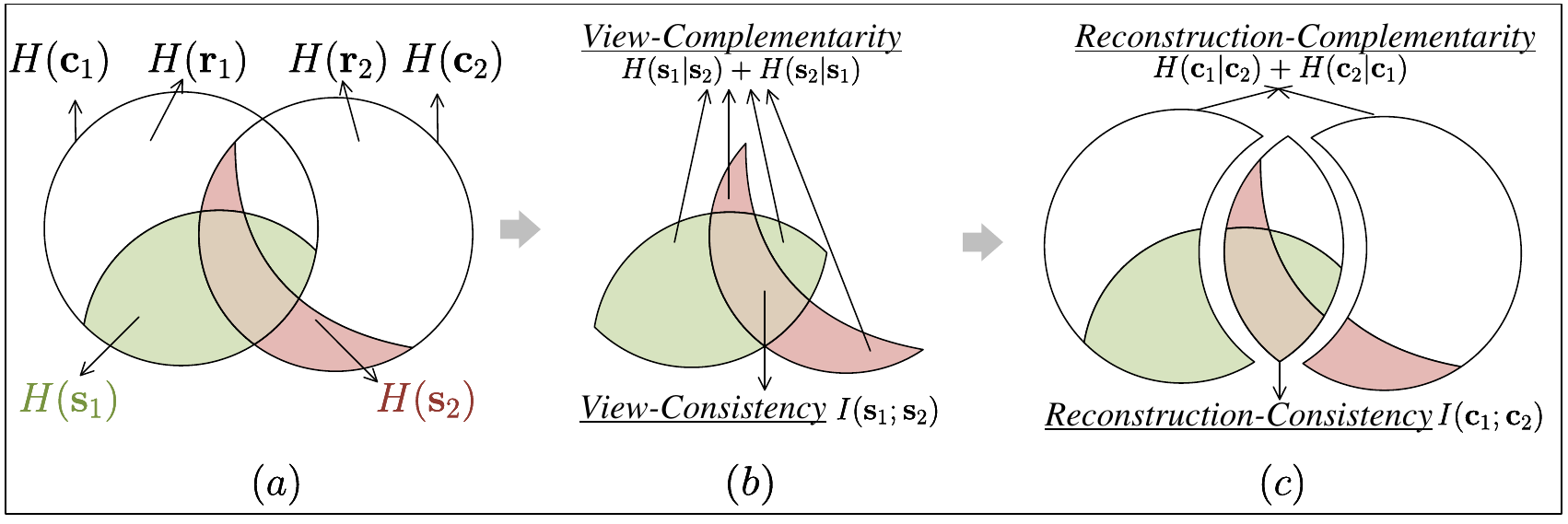}
	\caption{Venn diagram visualization of entropies and MI for six variables: $\mathbf{s}_1$, $\mathbf{s}_2$, $\mathbf{c}_1$, $\mathbf{c}_2$, $\mathbf{r}_1$ and $\mathbf{r}_2$ respectively, where $\mathbf{c}_1$, $\mathbf{c}_2$, $\mathbf{r}_1$ and $\mathbf{r}_2$ are defined by Definition \ref{Reconstruction-relevant information}.}
	\label{Analysis1_fig}
\end{figure}

\begin{definition}[Reconstruction-relevant information]
	\label{Reconstruction-relevant information}
	We formally define the reconstruction-relevant information as follows
	\begin{equation}
		H(\mathbf{c}) = H(\mathbf{s}) + H(\mathbf{r}),
	\end{equation}
	where $\mathbf{c} \sim p(\mathbf{c})$ is a reconstruction-relevant random variable. $H(\mathbf{c})$ is the entropy of $\mathbf{c}$, which 
	contains two parts: sufficient statistics of the source $H(\mathbf{s})$ and optimal error-resilient information $H(\mathbf{r})$ \footnote{Here, $H(\mathbf{s})$ acts as a proxy to measure the semantic content that ideally should be preserved. It provides a conceptual explanation for how the semantic information is distributed across different views, and how the reconstruction performance can benefit from both complementary and consistent information. }.
\end{definition}

$H(\mathbf{c})$ is an ideal lossless transmission limit, as illustrated in Fig. \ref{Analysis1_fig}. $I(\mathbf{s}_1;\mathbf{s}_2)$ and $(H(\mathbf{s}_1|\mathbf{s}_2)+H(\mathbf{s}_2|\mathbf{s}_1))$ are view-consistency and view-complementarity information, respectively, derived from the intrinsic overlap and differences between two views \cite{Li2022TKD}, whereas $I(\mathbf{c}_1;\mathbf{c}_2)$ and $(H(\mathbf{c}_1|\mathbf{c}_2)+H(\mathbf{c}_2|\mathbf{c}_1))$ extend consistency and complementarity to the reconstruction level, additionally expressing the uncertainty brought by the channel. The best balance between complementarity and consistency leads to optimal reconstruction performance.

\subsection{Probabilistic Model Analysis}From the variational perspective, DJSCC can be deemed as a variational autoencoder (VAE) \cite{Dai2022JSAC}. Then, we build a variational model for distributed image transmission. 

Specifically, the distribution of latent representation $q_{\mathbf{z}|\mathbf{s}}(\mathbf{z}|\mathbf{s})$ is learned by a transform $f(\cdot;\boldsymbol \phi)$ at the encoder, resulting in $\mathbf{z}_{1}=f(\mathbf{s}_1;\boldsymbol \phi)$. Similarly, the latent representation of the other view can be expressed as
$\mathbf{z}_{2}=f(\mathbf{s}_{2};\boldsymbol \phi)$. Finally, the reconstructed image can be derived from a DL-based decoder as ${\mathbf{\hat s}_1},{\mathbf{\hat s}_2}=g({\mathbf{z}_{1}},\mathbf{z}_{2};\boldsymbol{\theta})$.
In this setting, we aim to approximate the joint distribution of the random variables as $p(\mathbf{s}_1,\mathbf{s}_2,\mathbf{z}_{1},\mathbf{z}_{2})$, 
%\begin{equation}
%	p(\mathbf{s},\mathbf{s}_{side},\mathbf{z},\mathbf{z}_{side})=p(\mathbf{z})p(\mathbf{z}_{side})p_\theta(\mathbf{x}\mid\mathbf{z},\mathbf{z}_{side};\boldsymbol{\theta}),
%\end{equation}
which is intractable. To obtain a tractable solution, a factored variational approximation of the posterior distribution is introduced as $q(\mathbf{z}_{1},\mathbf{z}_{2}\mid\mathbf{s}_1,\mathbf{s}_{2})$.

\begin{equation}
	\label{kl}	
	\begin{aligned} 
		\small & \mathbb{E}_{\mathbf{s}_1, \mathbf{s}_{2} \sim p(\mathbf{s}_1, \mathbf{s}_{2})} D_{\mathrm{KL}}[q(\mathbf{z}_{1},\mathbf{z}_{2} \mid \mathbf{s}_1, \mathbf{s}_{2}) \| p(\mathbf{z}_{1},\mathbf{z}_{2} \mid \mathbf{s}_1, \mathbf{s}_{2})] \\
		&=\mathbb{E}_{\mathbf{s}_1, \mathbf{s}_{2} \sim p(\mathbf{s}_1, \mathbf{s}_{2})} \mathbb{E}_{\mathbf{z}_{1},\mathbf{z}_{2} \sim q}\Big(\log q(\mathbf{z}_{1} \mid \mathbf{s}_1)q(\mathbf{z}_{2} \mid \mathbf{s}_{2})\\
		&-\big(\underbrace{\log p(\mathbf{s}_1 \mid \mathbf{z}_{2}, \mathbf{z}_{1})}_{D_1}-\big(\underbrace{\log p(\mathbf{s}_2 \mid \mathbf{z}_{2}, \mathbf{z}_{1})}_{D_2} +\underbrace{\log p\left(\mathbf{z}_{1}\right)}_{R_1}\\
		&+\underbrace{\log p\left(\mathbf{z}_{2}\right)}_{R_2}\big)\Big)+ \text {const.}
	\end{aligned}
\end{equation}

We minimize the Kullback-Leibler (KL) divergence between the approximate density $q(\mathbf{z}_{1},\mathbf{z}_{2} \mid \mathbf{s}_1, \mathbf{s}_{2})$ 
and the true posterior $p(\mathbf{z}_{1},\mathbf{z}_{2} \mid \mathbf{s}_1, \mathbf{s}_{2})$ as Eq. (\ref{kl}). 
%Assuming the posterior follows normal distribution, 
The first term in the KL divergence can be technically dropped \cite{Dai2022JSAC}. The terms $D_1$, $D_2$ denote the reconstruction distortion. $R_1$ and $R_2$ denote the compression ratio of $\mathbf{z}_{1}$ and $\mathbf{z}_{2}$. 
% Our goal is to analyze the impact of the encoding and decoding process on the MI of the multi-view sources, for better utilizing the MI between the correlated sources base on a fixed CR, i.e., $R_{\mathbf{z}_{1}}$ and $R_{\mathbf{z}_{2}}$ can be viewed as constant. 

When considering a fixed compression ratio transmission, where $R_1$ and $R_2$ are constant,  minimizing the above KL divergence is equivalent with the following optimization problem,
% \begin{subequations}
%     \begin{align}
%         \min_{\boldsymbol{\phi},\boldsymbol{\theta}} \quad &  \mathbb{E}_{\mathbf{s}_1}\left[d\left(\mathbf{s}_1, \mathbf{\hat{s}_1}\right)\right] + \mathbb{E}_{\mathbf{s}_2}\left[d\left(\mathbf{s}_2, \mathbf{\hat{s}_2}\right)\right] \label{eq:obj} \\
%         \text{s.t.} \quad & R_1 = B_1, \label{eq:const1} \\
%                                 & R_2 = B_2, \label{eq:const2} \\
%                                 & P_1 = P_{\mathrm{total}_1}, \label{eq:const3} \\
%                                 & P_2 = P_{\mathrm{total}_2}, \label{eq:const4}
%     \end{align}
% \end{subequations}
\begin{subequations}\label{eq:obj}
\begin{align}
&\min_{\boldsymbol{\phi},\boldsymbol{\theta}}  && \mathbb{E}_{\mathbf{s}_1}\left[d\left(\mathbf{s}_1, \mathbf{\hat{s}_1}\right)\right] + \mathbb{E}_{\mathbf{s}_2}\left[d\left(\mathbf{s}_2, \mathbf{\hat{s}_2}\right)\right]    \tag{\ref{eq:obj}} \\
&\;\text{s.t.} && R_1 = B_1, \label{eq:const1}  \\
&             && R_2 = B_2, \label{eq:const2} \\
&             && P_1 = P_{\mathrm{total}_1}, \label{eq:const3} \\
&             && P_2 = P_{\mathrm{total}_2}, \label{eq:const4}
\end{align}
\end{subequations}
where $d(\cdot)$ denotes the mean square error (MSE). The objective function Eq. (\ref{eq:obj}) is the sum of reconstruction MSE from two views. $\boldsymbol{\theta}$ and $\boldsymbol{\phi}$ can be optimized using DL methods based on gradient descent with respective constraints, i.e., bandwidth constraints Eq. (\ref{eq:const1}) (\ref{eq:const2}) and power constraints Eq. (\ref{eq:const3}) (\ref{eq:const4}) for each view. 

In DJSCC, a squared error loss is equivalent to assuming a Gaussian likelihood for the reconstruction \cite{Balle2018arXiv}, i.e., 
\begin{equation}
    p_{\mathbf{s}\mid\mathbf{z}}(\mathbf{s}\mid\mathbf{z},\boldsymbol{\theta})=\quad\mathcal{N}\big(\mathbf{s}\mid\mathbf{\hat{s}},\mathbf{\Sigma}\big)~~\mathrm{with~}\mathbf{\hat{s}}=g(\mathbf{z};\boldsymbol{\theta}),
\end{equation}  
where $\mathbf{\Sigma}$ is the covariance matrix of Gaussian distribution. Here, the Gaussian assumption is imposed on the image domain $\mathbf{s}$ for analytical convenience, rather than the encoded latent domain $\mathbf{z}$. Therefore, from a training perspective, DJSCC is essentially an optimization processes based on a maximum likelihood estimation (MLE) under a fixed compression ratio \cite{Bourtsoulatze2019TCCN,Zhang2023TWC,Xu2022TCSVT,Wu2022WCL,Yang2022TCCN,Shao2023WCL}. However, from an analytical perspective, distributed DJSCC naturally aligns with a Bayesian interpretation, where $p(\mathbf{x}_1)$ serves as the prior for the transmission of $\mathbf{x}_1$, and $p(\mathbf{x}_1|\mathbf{z}_2)$ represents the posterior refined through cross-view observation \footnote{Observing one view (e.g., $\mathbf{z}_2$) helps reduce the uncertainty about the other (e.g., $\mathbf{x}_1$), and vice versa.}. Due to view disparity, channel distortion, and imperfect CSI, such observations are often noisy, and thus an adaptive compensation is needed. Next, we analyse how CSI estimation error affects performance within a Bayesian framework.

%\begin{figure}[htbp]
%	\centering
%	\includegraphics[height=8cm, width=8cm] {13-city_3d.png}
%	\caption{(a).cross attention mechanism. (b).cross-view information extraction.}
%	\label{TWC_encoder_cor}
%\end{figure}
 
\begin{figure*}[!tb]
\hrule
\begin{align}
&{\mu_{\mathcal{X}_1 \mid \mathcal{Z}_1, \mathcal{Z}_2} = \mu_{x_1} + \frac{ \mathcal{H}_1 \sigma_{x_1}^2 (\mathcal{H}_2^2 \sigma_{x_2}^2 (1 - r^2) + \sigma_{\tilde w}^2) ({\mathcal{Z}_1 - \mathcal{H}_1 \mu_{x_1}}) + \mathcal{H}_2 r \sigma_{x_1} \sigma_{x_2} \sigma_{\tilde w}^2 ({\mathcal{Z}_2 - \mathcal{H}_2 \mu_{x_2} })}{ \mathcal{H}_1^2 \mathcal{H}_2^2 \sigma_{x_1}^2 \sigma_{x_2}^2 (1 - r^2) + \mathcal{H}_1^2 \sigma_{x_1}^2 \sigma_{\tilde w}^2 + \mathcal{H}_2^2 \sigma_{x_2}^2 \sigma_{\tilde w}^2 + \sigma_{\tilde w}^4}.}\label{post1}\\
&{\sigma_{\mathcal{X}_1 \mid \mathcal{Z}_1, \mathcal{Z}_2}^2 = \sigma_{x_1}^2 - \frac{ 
\mathcal{H}_1^2 \sigma_{x_1}^4 \mathcal{H}_2^2 \sigma_{x_2}^2 (1 - r^2) + \mathcal{H}_1^2 \sigma_{x_1}^4 \sigma_{\tilde w}^2 + \mathcal{H}_2^2 r^2 \sigma_{x_1}^2 \sigma_{x_2}^2 \sigma_{\tilde w}^2 
}{ 
\mathcal{H}_1^2 \mathcal{H}_2^2 \sigma_{x_1}^2 \sigma_{x_2}^2 (1 - r^2) + \mathcal{H}_1^2 \sigma_{x_1}^2 \sigma_{\tilde w}^2 + \mathcal{H}_2^2 \sigma_{x_2}^2 \sigma_{\tilde w}^2 + \sigma_{\tilde w}^4. 
}.}\label{post2}
\end{align}
% \hrule
\end{figure*}
\subsection{Imperfect CSI  Analysis}\label{Imperfect CSI}
Different from the perfect CSI setting in \cite{Wang2022ICASSP}, we analyse the impact of CSI estimation error under both data-driven and model-driven CSI estimation approaches \cite{Yang2022TCCN}. Data-driven approach relies on the DNN to learn the underlying information about CSI from the channel output pilots and signal. It treats signal processing as a black box. Model-driven approach relies on domain knowledge to guide the design of DNN. For example, we can use minimum mean square error (MMSE) or least square (LS) estimator to estimate CSI \footnote{It can be proved that model-driven approaches generally provide more accurate CSI estimation compared to data-driven methods, but this improvement partly comes at the cost of increased neural network complexity \cite{Yang2022TCCN}.}. \\

However, both data-driven and model-driven approaches inevitably introduce CSI estimation error. Considering a specific subcarrier, the frequency-domain CSI $\mathcal H$ and the associated CSI estimation error can be modeled as
\begin{equation}
        {{\mathcal H} = \hat{\mathcal H} + \mathcal{E}},\label{CSI_error}
\end{equation}
where $\hat{\mathcal{H}}$ denotes the estimated CSI, and $\mathcal{E}\sim \mathcal{N}(\mathcal{E};0, \sigma_{e}^2)$ represents the corresponding CSI estimation error. Next, we investigate how CSI estimation error affects the posterior estimation and correlation within a Bayesian framework.

In a Bayesian framework, we treat the transmitted signal $\mathbf{x}$ as a random vector sampled from a statistical distribution $p(\mathbf{x})$. The goal is to estimate $\mathbf{x}$ by maximizing the posterior probability given the received observation $\mathbf{z}$. This yields the MAP estimator
\begin{equation}
    {\hat{\mathbf{x}}^{\text{MAP}} = \arg\max_{\mathbf{x}} p(\mathbf{x}|\mathbf{z})= \arg\max_{\mathbf{x}}p({\mathbf{z}}|\mathbf{x})p({\mathbf{x}}),}\label{map1}
\end{equation}
where the likelihood function $p({\mathbf{z}}|\mathbf{x})$ is determined by channel model. Exact characterization of $p(\mathbf{x})$ is non-trivial. Hence, we assume that the elements of $\mathbf{x}$ are drawn from a generic Gaussian random variable in an i.i.d. way  \cite[Remark 2]{Shao2025}. Formally, we define the Gaussian random variable $\mathcal{X} \sim \mathcal{N}(\mathcal{X};\mu_{x}, \sigma_{x}^2)$, where $\mu_{x}$ and $\sigma_{x}^2$ denote the sample mean and sample variance computed from the observed realizations of $\mathbf{x}$  \footnote{We emphasize that there is no explicit prior distribution on the latent representations $\mathbf{x}$ in practice. Here, Gaussian modeling of the latent representations is adopted solely for analytical convenience within a Bayesian framework}. Further, the received observation $\mathbf{z}$ is also drawn i.i.d., and can be viewed as realization of a random variable
\begin{equation}
    {\mathcal{Z} = \mathcal{H} \mathcal{X} + \mathcal{W}},\label{ch1}
\end{equation}
where $\mathcal{W} \sim \mathcal{N}(\mathcal{W};0, \sigma_{w}^2)$ denotes AWGN. Correspondingly, Eq. \eqref{map1} can be written as
\begin{equation}
    {\hat{\mathcal{X}}^{\text{MAP}} = \arg\max_{\mathcal{X}} p(\mathcal{X}|\mathcal{Z})= \arg\max_{\mathcal{X}}p({\mathcal{Z}}|\mathcal{X})p({\mathcal{X}}).}\label{map2}
\end{equation}
For the sake of further analysis, we make the following assumptions. 
\begin{assumption}
	\label{assumption2}
	In distributed DJSCC, there exist two transmitted signals from two views, i.e., $\mathcal{X}_1 \sim \mathcal{N}(\mathcal{X}_1;\mu_{x_1}, \sigma_{x_1}^2)$, $\mathcal{X}_2 \sim \mathcal{N}(\mathcal{X}_2;\mu_{x_2}, \sigma_{x_2}^2)$. Given the statistical dependency between the two views, it is natural to assume that $\mathcal{X}_1$ and $\mathcal{X}_2$ follow a joint Gaussian distribution, i.e., $\mathcal{X}_1, \mathcal{X}_2 \sim \mathcal{N}(\mu_{x_1}, \mu_{x_2}, \sigma_{x_1}^2, \sigma_{x_2}^2, r)$, where $r$ denotes the correlation coefficient.
\end{assumption}

 Then, Eq. \eqref{map2} can be written as \footnote{Due to the symmetry between the two views, without loss of generality, we focus on one view for detailed analysis, i.e., $\mathcal{X}_1$.}
\begin{equation}
\begin{aligned}
        {\hat{\mathcal{X}_1}^{\text{MAP}}} &{= \arg\max_{\mathcal{X}_1} p(\mathcal{X}_1|\mathcal{Z}_1,\mathcal{Z}_2)}\\     &{=\arg\max_{\mathcal{X}_1}p({\mathcal{Z}_1,\mathcal{Z}_2}|\mathcal{X}_1)p({\mathcal{X}_1})},\label{map3}
\end{aligned}
\end{equation}
where the likelihood function $p({\mathcal{Z}_1,\mathcal{Z}_2}|\mathcal{X}_1)$ is determined by Eq. \eqref{CSI_error} and Eq. \eqref{ch1}. By substituting Eq. \eqref{CSI_error} into Eq. \eqref{ch1}, we have
\begin{equation}\label{received1}
{{\mathcal{Z}} = (\hat{\mathcal H} + \mathcal{E}) \mathcal{X} + \mathcal{W}=\hat{\mathcal H}\mathcal{X} + \underbrace{\mathcal{E}\mathcal{X}  + \mathcal{W}}_{\text{Equivalent noise} \ \widetilde{\mathcal{W}}}.} 
\end{equation} 
It can be observed that CSI estimation error  introduces additional uncertainty into the likelihood function. The equivalent noise term is defined as
\begin{equation}\label{equivalent_noise}
    {\widetilde{\mathcal{W}}\sim \mathcal{N}(\widetilde{\mathcal{W}};0, \sigma_{\tilde w}^2)~~\mathrm{with~}\sigma_{\tilde w}^2 = \sigma_{e}^2\sigma_{x}^2+ \sigma_{w}^2.}
\end{equation}
From an information-theoretic perspective, the CSI estimation error  would result in an SNR loss and thus degrade capacity \cite[Eq.(12)]{Yoo2006}.

{Next, we can derive the Bayesian estimation. According to Bayes’ theorem, we have }
\begin{equation}
    {p(\mathcal{X}_1|\mathcal{Z}_1,\mathcal{Z}_2) \propto p({\mathcal{Z}_1,\mathcal{Z}_2}|\mathcal{X}_1)p({\mathcal{X}_1})}
\end{equation}
{where the multiplication of two Gaussians is still a Gaussian, thus }
\begin{equation}
{ p(\mathcal{X}_1|\mathcal{Z}_1,\mathcal{Z}_2) \sim \mathcal{N}(\mathcal{X}_1;\mu_{\mathcal{X}_1 \mid \mathcal{Z}_1, \mathcal{Z}_2}, \sigma_{\mathcal{X}_1 \mid \mathcal{Z}_1,\mathcal{Z}_2}^2), }  
\end{equation}
{where $\mu_{\mathcal{X}_1 \mid \mathcal{Z}_1, \mathcal{Z}_2}$ and $\sigma_{\mathcal{X}_1 \mid \mathcal{Z}_1, \mathcal{Z}_2}^2$ are given in} Eq. \eqref{post1} and Eq. \eqref{post2}. {Let us take a closer look:}
\begin{enumerate}
\item  {\textbf{Cross-View Correction}: The posterior mean $\mu_{\mathcal{X}_1 \mid \mathcal{Z}_1, \mathcal{Z}_2}$ fuses information from both views ($\mathcal{Z}_1$ and $\mathcal{Z}_2$). The term $\mathcal{H}_2 r \sigma_{x_1} \sigma_{x_2} \sigma_{\tilde w}^2 (\mathcal{Z}_2 - \mathcal{H}_2 \mu_{x_2})$ explicitly uses $\mathcal{X}_2$ to correct the estimation of $\mathcal{X}_1$.}
\item {\textbf{Equivalent Noise Suppression}: The posterior variance $\sigma_{\mathcal{X}_1 \mid \mathcal{Z}_1, \mathcal{Z}_2}^2$ decreases as equivalent noise $\sigma_{\tilde w}^2$ decreases, which means that lower CSI estimation error  $\sigma_e^2$ or channel noise $\sigma_w^2$ directly reduces uncertainty. }
\item {\textbf{Extreme Cases}}: \textcircled{1}$r \to 0$, {it means no correlation between views \footnote{{In this case, view $\mathcal{X}_2$ provides no useful information for reconstructing view $\mathcal{X}_1$, and thus its transmission becomes unnecessary, i.e., $\mathcal{H}_2 = 0$.}}. The posterior variance degenerates to $
\sigma_{\mathcal{X}_1 \mid \mathcal{Z}_1, \mathcal{Z}_2}^2 = \frac{\sigma_{x_1}^2 \sigma_{\tilde w}^2}{\mathcal{H}_1^2 \sigma_{x_1}^2 + \sigma_{\tilde w}^2}$, 
which corresponds to the classical MMSE estimation} \cite[Chapter 10]{Kay1993}. \textcircled{2}$\sigma_{e} \to \infty$, {thus $\mu_{\mathcal{X}_1|\mathcal{Z}_1,\mathcal{Z}_2} \to \mu_{x_1}$ and $\sigma_{\mathcal{X}_1|\mathcal{Z}_1,\mathcal{Z}_2}^2 \to \sigma_{x_1}^2$. It means that the posterior naturally falls back to prior, avoiding overfitting to noisy observations when CSI is highly unreliable. }
\end{enumerate}

{Based on the above analysis, we could find that inter-view correlation has a direct impact on the posterior. Meanwhile, CSI estimation error also degrades the statistical correlation between the received signals at the decoder (CSI estimation error further degrades the equalization performance), which makes such correlation increasingly intractable. }
\begin{remark}
	\label{remark3}
{The correlation among views is difficult to compute, especially in continuous and high-dimensional encoded representations. This challenge is exacerbated by CSI estimation errors, which distort inter-view dependencies. To address this, we analyze the MI of multi-view encoded representations under imperfect CSI, providing guidance for optimizing consistency and complementarity at the reconstruction level.}

\end{remark}

\subsection{Mutual Information Analysis}\label{sec:Correlation}
In this section, we first analyze the MI changes of multi-view encoded representations during the semantic encoding process. Then, we analyze the influence of noisy channel.
%Theorem \ref{MI Reduction} shows the the MI between the two encoded representations would become weaker over noisy channel. Further explanation, SCS between encoded representations is given in Fig. \ref{cor analysis}(a). Based on assumption (\ref{assumption2}), MI between encoded representations also reduced as shown in Fig. \ref{cor analysis}(b). Meanwhile, we can observe that SCS decreased when two encoded representations x1\mathbf{x}_1 and x2\mathbf{x}_2 over the channel. Meanwhile, SCS of ˆx1\mathbf{z}_{1} and ˆx2\mathbf{z}_{2} grows as channel conditions get better. This phenomenon is consistent with Eqn. (\ref{noisy correlation coefficient}). Of course, we also observe that SCS of x1\mathbf{x}_1 and x2\mathbf{x}_2 decreases as SNR rises. This is because we introduce SNR as a prior in the encoder. SNR prior could help encoder further distinguish the differences of correlated sources. 
\begin{figure}[htbp]
	\centering
	\includegraphics[height=2.5cm, width=8cm] {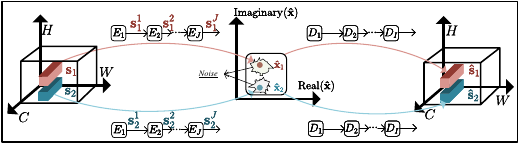}
	\caption{Geometrical interpretation of the mappings performed by the network in Fig. \ref{intro5}(b). Two view image sources $\mathbf{s}_{1}\in\mathbb{R}^{C \times H\times W}$, $\mathbf{s}_{2}\in\mathbb{R}^{C \times H\times W}$ are encoded by a semantic encoder, and recovered by a semantic decoder.  }
	\label{TWC_encoder_cor}
\end{figure}
\subsubsection{Mutual Information Analysis under Semantic Encoding} Analytically calculating the MI of multi-view encoded representations is non-trivial due to the high nonlinearity of the DNN. Thus, we resort to an information-theoretic analysis for a further insight.
The encoder for DJSCC based on convolutional backbone networks is composed of multiple stacked convolutional blocks or residual blocks, which can be expressed as
\begin{equation}
	\begin{aligned}
	f(\cdot; \boldsymbol{\phi}):=\bigcup\limits_{j=1}^JE_j= E_1\circ E_2\cdots\circ E_J,\\
	g(\cdot; \boldsymbol{\theta}):=\bigcup\limits_{i=1}^ID_i= D_1\circ D_2\cdots\circ D_I,
	\end{aligned}
\end{equation}
where $E_j$, $D_i$ respectively denote the $j$-$th$, $i$-$th$ stacked blocks of codecs respectively, and $\circ$ denotes the stacking of blocks. As shown in Fig. (\ref{TWC_encoder_cor}), $\mathbf{s}_1^j$ and $\mathbf{s}_2^j$ are the output of the view $\mathbf{s}_1$ and $\mathbf{s}_2$ at the $j$-$th$ block.
By the data processing inequality, we have 
\begin{equation}
	\begin{aligned}
		I({\mathbf{s}}_1; {\mathbf{s}}_1^a) \geq I({\mathbf{s}}_1; {\mathbf{s}}_1^b), 1 \leq a \leq b \leq J,\\
		I({\mathbf{s}}_2; {\mathbf{s}}_2^a) \geq I({\mathbf{s}}_2; {\mathbf{s}}_2^b), 1 \leq a \leq b \leq J,
	\end{aligned}
\end{equation}
{Meanwhile, the Markov chain $\mathbf{s}_1 \to \mathbf{s}_2 \to \mathbf{s}_2^j$ holds, since $\mathbf{s}_2^j$ is generated solely from $\mathbf{s}_2$ by the encoder. This implies that $\mathbf{s}_2^j$ is conditionally independent of $\mathbf{s}_1$ given $\mathbf{s}_2$. By symmetry, we also have $\mathbf{s}_2 \to \mathbf{s}_1 \to \mathbf{s}_1^j$.}

\begin{theorem}[Mutual Information Non-Increasing Theorem]
During the same encoding process across multiple stages, the MI between the two correlated views is non-increasing. Formally, for any $1 \leq a \leq b \leq J$, we have
\label{MI Reduction}
\begin{equation}\label{MI_reduction}
    I({\mathbf{s}}_1^a; {\mathbf{s}}_2^a) \geq I(\mathbf{s}_1^b; \mathbf{s}_2^b),
\end{equation}
\end{theorem}

\begin{proof}
Let $p_{{\mathbf{s}_1^a}, {\mathbf{s}_2^a}} = p_1^a$, $p_{\mathbf{s}_1^a}p_{\mathbf{s}_2^a} = p_2^a$, $p_{{\mathbf{s}_1^b}, {\mathbf{s}_2^b}} = p_1^b$, $p_{\mathbf{s}_1^b}p_{\mathbf{s}_2^b} = p_2^b$. By the data-processing inequality applied to relative entropies (see \cite{Thomas2006} pp. 370–371), we have 
\begin{equation}
    D_{KL}(p_{1}^a||p_{2}^a) \geq D_{KL}({p}_{1}^b||{p}_{2}^b)
\end{equation}
According to the definition of MI $I(\mathbf{s}_1^a;\mathbf{s}_2^a)=D(p_{\mathbf{s}_1^a}\parallel p_{\mathbf{s}_2^a})$, {Theorem~\ref{MI Reduction} is proved.}
\end{proof}

\begin{remark}
	\label{remark2}
{From a VAE perspective, the encoder and decoder play game-like roles: the encoder compresses the input, while the decoder reconstructs. In this game-like structure, the decoder attempts to preserve as much reconstruction-relevant information as possible by learning robust mappings from latent representations, even under noisy channels and imperfect CSI conditions. This inspires our decoder-side CVIE and CCF modules in Section \ref{sec:Decoding}, which exploits cross-view consistency and complementarity. }
\end{remark}
%It is necessary to design flexible multi-view signal optimization based on the reconstruction-relevant MI to meet the requirements of consistency and complementarity in various scenarios. 
%Nevertheless, MI is notoriously difficult to compute, particularly in continuous and high-dimensional encoded representations. Here, we use squared cosine similarity (SCS) to roughly describe the correlation between two representations $\mathbf{x}_1, \mathbf{x}_2$ \cite{WF2023WCL,Rahutomo2012},
%\begin{equation}
%	\label{SCS_eq}
%	\mathrm{cos}^2(\mathbf{x}_1,\mathbf{x}_2)\triangleq\frac{\langle\mathbf{x}_1,\mathbf{x}_2\rangle^2}{\|\mathbf{x}_1\|^2\|\mathbf{x}_2\|^2}
%\end{equation}

\subsubsection{Mutual Information Analysis under Noisy Channel} \label{sec:Analysis3}
To eliminate the impact of noise on the correlated sources, we should first understand the impact of noise. 
\begin{assumption}
	\label{assumption1}    
{From the encoding perspective, $\mathbf{z}_1$ and $\mathbf{z}_2$ represent the noisy channel outputs corresponding to different views. A well-designed encoder is expected to retain at least partial MI with the original source. Unless the channel is completely destructive, it is naturally and empirically assumed that $\mathbf{x}_1$ and $\mathbf{z}_1$ are not statistically independent, i.e., $I(\mathbf{x}_1; \mathbf{z}_1) > 0$. Furthermore, in practical scenarios (e.g., stereo vision, multi-camera surveillance, or distributed sensing), the encoded representations $\mathbf{x}_1$ and $\mathbf{x}_2$ are statistically correlated. Given that the channel is not completely destructive, it is empirically reasonable to assume that $\mathbf{z}_2$, as the noisy observation of $\mathbf{x}_2$, retains partial information about $\mathbf{x}_1$. Thus, we assume $I(\mathbf{x}_1; \mathbf{z}_2) > 0$.
}
\end{assumption}

%$\mathbf{x}_1$ and $\mathbf{z}_1$ are two dependent random variables, i.e., $I(\mathbf{x}_1;\mathbf{z}_1)\textgreater 0$. 
For the transmission of $\mathbf{x}_1$, we have $I(\mathbf{x}_{1};\mathbf{z}_1|\mathbf{x}_{2})\leq I(\mathbf{x}_{1};\mathbf{z}_1)$ since the decoder can use the correlated information $\mathbf{x}_2$. In practice, $\mathbf{x}_2$ is often not lossless (recall Remark \ref{remark1}). Hence, the actual optimization goal is,  $I(\mathbf{x}_{1};\mathbf{z}_1|\mathbf{z}_2)$. Based on Assumption \ref{assumption1}, we have $I(\mathbf{x}_{1};\mathbf{z}_1|\mathbf{z}_2) < I(\mathbf{x}_{1};\mathbf{z}_1)$, which ensures that the noisy correlated sources can still help decode. In addition, Theorem \ref{MI Reduction} still holds when multi-view encoded representations are transmitted over a noisy channel as we have
	\begin{equation} 
	I({\mathbf{x}}_1; {\mathbf{x}}_2) \geq I(\mathbf{z}_{1}; \mathbf{z}_{2}).
\end{equation}
{For the sake of further analysis, we adopt a simplified one-dimensional Gaussian model based on Assumption~\ref{assumption2}. Specifically}, following Eq. \eqref{ch1}, the received signals can be modeled as
\begin{equation}
{\mathcal{Z}_1 = \mathcal{H}_1 \mathcal{X}_1 + \mathcal{W}_1, \quad \mathcal{Z}_2 = \mathcal{H}_2 \mathcal{X}_2 + \mathcal{W}_2},
\end{equation}
{Assuming that $\mathcal{H}_1$ and $\mathcal{H}_2$ are estimated by} Eq. \eqref{CSI_error}, {the MI between the received signals $\mathcal{Z}_1$ and $\mathcal{Z}_2$ can be derived as:}
\begin{equation}
{I(\mathcal{Z}_1;\mathcal{Z}_2) = -\frac{1}{2}\log(1 - r'^2)},
\end{equation}
{where $r'$ denotes the correlation coefficient between $\mathcal{Z}_1$ and $\mathcal{Z}_2$, which can be expressed as}
\begin{equation}
{r' = \frac{r \sigma_{x_1} \sigma_{x_2}}{
\sqrt{
\left( \sigma_{x_1}^2 + \frac{\sigma_{\tilde w_1}^2}{\hat{\mathcal{H}}_1^2} \right)
\left( \sigma_{x_2}^2 + \frac{\sigma_{\tilde w_2}^2}{\hat{\mathcal{H}}_2^2} \right)
}},}\label{noisy_correlation}
\end{equation}
{where $\sigma_{\tilde w_1}^2$ and $\sigma_{\tilde w_2}^2$ denote the equivalent noise from two views respectively,} as Eq. \eqref{equivalent_noise}. {Obviously, we have $0 \leq r' \leq r$ and thus $I(\mathcal{Z}_1; \mathcal{Z}_2) \leq I(\mathcal{X}_1,\mathcal{X}_2)$.}

According to the above analysis, we notice that the reconstruction consistency of multi-view encoded representations deteriorates over the noisy channel. 
{Hence, it is necessary to design a flexible multi-view transmission framework to meet the requirements of consistency and complementarity in a distributed image transmission scenario. This is also aligned with Remark~\ref{remark3}.}

\section{Method}\label{sec:Decoding}
Guided by the analysis in Section \ref{sec:Analysis}, we propose a distributed wireless image transmission scheme in this section. The scheme
is composed of a novel cross-view information extraction (CVIE) mechanism and a complementarity-consistency fusion (CCF) mechanism.

\subsection{Cross-View Information Extraction}
{Remark~\ref{remark2} motivates the placement of CVIE modules at the decoder side.} Specifically, 
CVIE uses features generated by three individual $1 \times 1$ convolution as inputs to perform cross-view interactions. The objective of CVIE is to learn cross-view information. CVIE is inspired by CAM \cite{Wang2022ICASSP,zhang2023ICLR}, as shown in Fig. \ref{TWC_acmix1}(a). CAM uses $1 \times 1$ convolutions $W_q$, $W_k$, $W_v$ to map the view to its query, key and value as
\begin{equation}
	\label{Cross-attention-mechanism}   
 q_{ij}^2=W_q\mathbf{z}_{ij}^2,k_{ij}^1=W_k\mathbf{z}_{ij}^1,\upsilon_{ij}^1=W_v\mathbf{z}_{ij}^1,
\end{equation}
where $q_{ij}^2$ denotes the query from the second view at a local region of pixels $(i,j)$,  $k_{ij}^1$ and  $v_{ij}^1$ are from the first view. Like widely adopted self-attention mechanism, CAM can be formulized as
\begin{equation}
	\label{Cross-attention-mechanism} 
	\mathrm{Attention}(q_{ij}^2,k_{ij}^1,v_{ij}^1)=\mathrm{softmax}({q_{ij}^2k_{ij}^1})v_{ij}^1, 
\end{equation}
where ${q_{ij}^2k_{ij}^1}$ denotes the similarity score of two views, the $\mathrm{softmax}$ function is applied to convert the unnormalized similarity score into similarity weights and these similarity weights are then used to compute a weighted sum of $v_{ij}^1$. 
\begin{figure}[htbp]
	\centering
	\includegraphics[height=2.8cm, width=8cm] {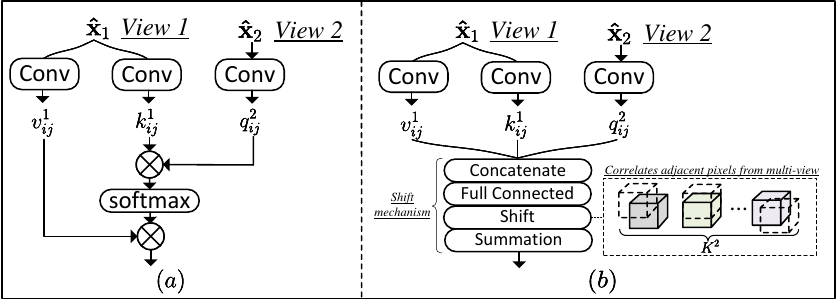}
	\caption{(a) Cross attention mechanism. (b) cross-view information extraction.}
	\label{TWC_acmix1}
\end{figure}

To capture more detailed cross-view patterns and dependencies, we propose CVIE shown in Fig. \ref{TWC_acmix1}(b). 
CVIE leverages the inherent principles of standard convolution for image coding. Standard convolution is a cross-correlation operation (see \cite{Zhang2021rXiv} pp. 246–247), which calculates the cross-correlation between input tensor $\mathbf{x}_{i,j}$ and a $K \times K$ kernel tensor $\boldsymbol{\mathcal{F}}$ as
\begin{equation}
	\label{Traditional convolution}	
	\begin{aligned}\mathbf{y}_{i,j}=\sum_{p,q}\boldsymbol{\mathcal{F}}_{p,q}\mathbf{x}_{i+p-\lfloor K/2\rfloor,j+q-\lfloor K/2\rfloor}\end{aligned},
\end{equation}
where the kernel weights $\boldsymbol{\mathcal{F}}_{p,q}$ at positions $p, q \in {0, 1, \ldots, K-1}$ characterize the weights associated with the kernel at the position $(p, q)$. 
% Here, naturally leads us to perform cross-correlations operation between two views. 
%However, directly performing cross-correlation on two high-dimensional variables is computationally expensive. 
Inspired by \cite{Pan2022CVPR}, which proves that standard convolution with kernel size $K \times K$ can be decomposed into $K^2$ individual 1 × 1 convolutions by the shift mechanism. We initially perform individual 1 × 1 convolutions from two views and obtain query, key and value as CAM. Next, a shift mechanism is adopted. Specifically, query, key and value are concatenated as
\begin{equation}
	\label{Cross-attention-mechanism} 
	\mathcal{Y}_{ij} = \mathrm{Concatenate}[q_{ij}^2,k_{ij}^1,v_{ij}^1], 
\end{equation}
where $\mathcal{Y}_{ij} \in \mathbb{R}^{3 \times C \times HW}$ is output,  $C$, $H$ and $W$ are the height, width, and channel of the output features, respectively. To fuse the information from multi-views, we utilize a multilayer perceptron (MLP) to generate $K^2$ multi-view features as $\mathcal{Y}_{ij}^{fc} \in \mathbb{R}^{K^2 \times C \times HW}$. Since these features are generated from 1 × 1 convolution, the correlation of adjacent pixels is not calculated. We resort to shift operation, which is the core of the shift mechanism. Shift operation can be expressed as
\begin{equation}
	\label{Shift} 
	\mathcal{Y}_{i,j}^{shift}=\mathrm{Shift}[\mathcal{Y}_{ij}^{fc},\Delta x, \Delta y]={\mathcal{Y}}_{i+\Delta x,j+\Delta y}^{fc},
\end{equation}
where $\Delta x, \Delta y$ correspond to the horizontal and vertical displacements, respectively. As shown in Fig. \ref{TWC_acmix1}(b), shifting features towards various directions to correlate adjacent pixels from multi-view, which further capture more nuanced cross-view patterns and dependencies. Finally, we fuse these $K^2$ features in each direction by summation operation as
\begin{equation}
	\label{Sum} 
	\mathcal{Y}_{i,j}^{Sum}=\sum_{p,q}\mathcal{Y}_{i,j}^{shift}.
\end{equation}

According to the above process, CVIE can compute cross-correlation between two views.

\begin{remark}
	\label{remark5}
{Unlike standard dot-product attention, our method adaptively fuses QKV features from different views via an MLP, guided implicitly by principles of complementarity and consistency. This data-driven fusion replaces explicit similarity scoring, enabling flexible and effective cross-view interaction.}
\end{remark}

\subsection{Complementarity-Consistency Fusion Mechanism} \label{sec:Fusion}
To further fuse the complementarity and consistency from multi-view information in a symmetric and compact manner, we further propose a CCF mechanism based on CVIE. As shown in Fig. \ref{TWC_acmix2}(a), we add a parallel path for single-view consistency information extraction on the basis of CVIE. The reused individual $1 \times 1$ convolution reduces computational overhead and improves the model capacity. Complementary and consistent information are weighted dynamically by $\mathcal{K}_1$, $\mathcal{K}_2$ and then summed. $\mathcal{K}_1$ and $\mathcal{K}_2$ are based on the dynamic weight assignment (DWA) policy shown in Fig. \ref{TWC_acmix2}(b). DWA is mainly composed of an MLP with activation function $\mathrm{ReLU}$ / $\mathrm{Softmax}$. {The input of DWA consists of SNRs of two independent OFDM channels $SNR_1$, $SNR_2$ and squared cosine similarity (SCS) of OFDM signals from two views computed as \cite{WF2023WCL,Nguyen2024PACC}} \footnote{{Since the cosine similarity is only employed as an auxiliary consistency measure, both the absolute value and the squared form SCS are theoretically valid and yield comparable performance in practice. 
}}
\begin{equation}
	\label{SCS_eq}
	\mathrm{cos}^2(\mathbf{x}_1,\mathbf{x}_2)\triangleq\frac{\langle\mathbf{x}_1,\mathbf{x}_2\rangle^2}{\|\mathbf{x}_1\|^2\|\mathbf{x}_2\|^2}.
\end{equation}
These three metrics are concatenated as Eq. (\ref{Cross-attention-mechanism}). The output of DWA can be modeled as a Bernoulli random variable $\mathcal{K} \sim \text{Bern}(p)$, where $\mathcal{K} \in \{\mathcal{K}_1, \mathcal{K}_2\}$.

\begin{figure}[htbp]
	\centering
	\includegraphics[height=2.8cm, width=8cm] {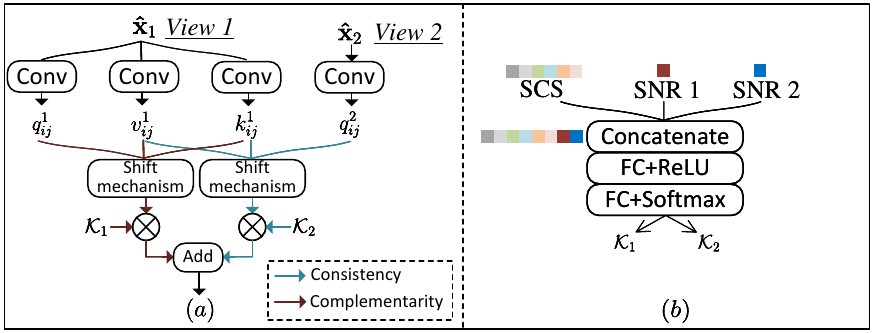}
	\caption{(a) Complementarity-consistency fusion mechanism. (b) Dynamic weight assignment policy.}
	\label{TWC_acmix2}
\end{figure}

The reasons why adopting SCS and SNR to balance the consistency and complementarity are as follows: First, as Remark~\ref{remark3}, the MI is notoriously difficult
to compute, particularly in continuous and high-dimensional
encoded representations. Secondly, CSI estimation errors can be modeled as an equivalent increase in noise power, as reflected in the received signal model in Eq. \eqref{received1}, which further degrades the effective receive SNR and correlation in Eq. \eqref{noisy_correlation}. Finally, the relationship between the SCS of two encoded representations and SNR is shown in Fig. \ref{cor analysis}. It can be noticed that a relatively high-quality channel corresponds to a relatively high SCS. Based on Eq. (\ref{noisy_correlation}), better channel conditions may lead to higher correlation and thus enhance the MI of two views under the assumption of jointly normal correlation (Assumption \ref{assumption2}). Therefore, Fig. \ref{cor analysis} indicates that the MI of received encoded representations between two views can be roughly predicted according to the SCS. Similar to Section \ref{sec:Correlation}, the further theoretical analysis for the impact of encoding process on SCS is given in Appendix \ref{sec:Impacts}.

\begin{figure}[htbp]
	\centering
	\label{Experiment 1}
	\includegraphics[height=4.5cm, width=5.5cm] {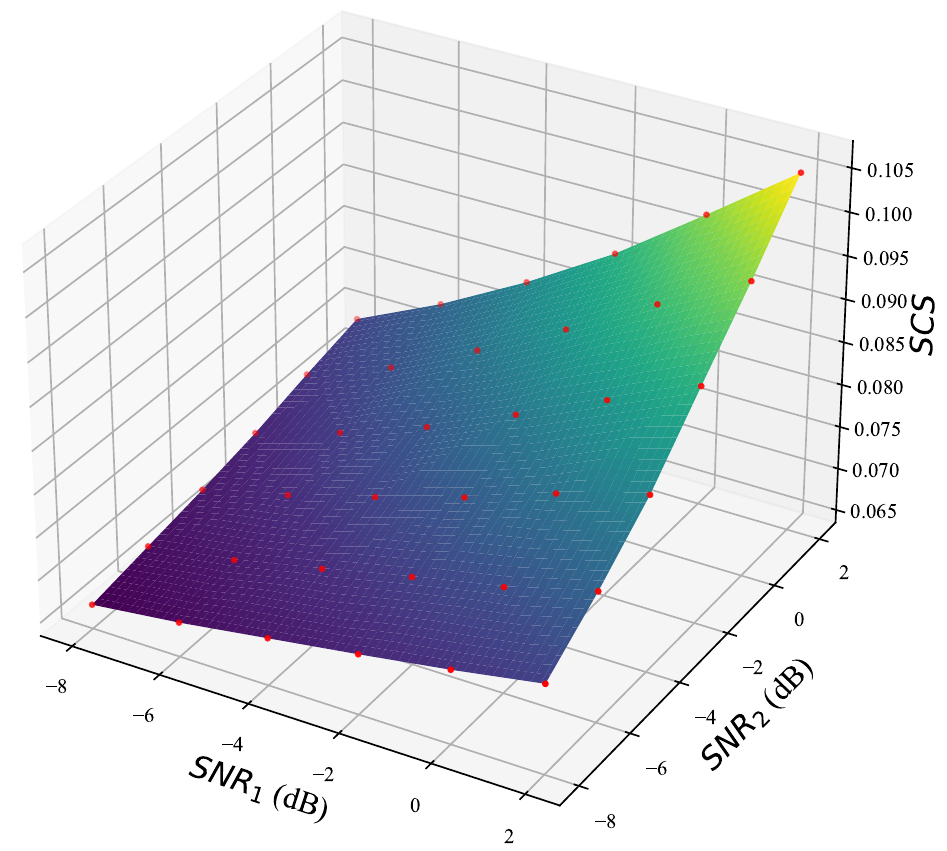}
	\caption{The squared cosine similarity of two received encoded representations under different SNRs.}
	\label{cor analysis}
\end{figure}

\section{Experiments} \label{sec:Experiments}
Next, we validate the performance of the proposed RDJSCC. 
First, we present the simulation settings. Then, we evaluate and analyze the performance.

\subsection{Datasets and Simulation Settings}
\subsubsection{Dataset} We quantify the distributed image transmission performance on {the} RGB Cityscapes dataset \cite{Cordts2016CVPR}, which is composed of stereo image pairs, and each pair is captured by a pair of cameras at the same moment. We adopt 2975 pairs for training, 500 pairs for validation, and 1525 pairs for testing. Each image of Cityscapes is downsampled to 128 × 256 pixels. In addition, KITTI dataset is used \cite{Geiger2012}. We adopt 1576 pairs for training, 790 pairs for validation, and 790 pairs for testing. Following \cite{Mital2022DCC,Wang2022ICASSP}, each image with 375 × 1242 pixels is centre-cropped and downsampled to 128 × 256 pixels.

\begin{table}[htbp]
	\caption{Simulation parameter settings.}
	\label{table0}
	\centering
         \small
	\begin{tabular}{|c|c|c|}
		\hline
		\textbf{}                                     & \textbf{Parameters} & \textbf{Value} \\ \hline
		\multirow{3}{*}{Channel environment} & $L$                   & 8              \\ \cline{2-3} 
		& $\gamma$                   & 4              \\ \cline{2-3} 
		& $SNR$                   & {[}-8,2{]} dB   \\ \hline
		\multirow{5}{*}{OFDM settings}                
		&$N_p$                 & 2               \\ \cline{2-3} 
		&$N_s$                   &  3              \\ \cline{2-3} 
		&$N_c$                  & 2048             \\ \cline{2-3}
            &$L_{\mathrm {cp}}$            & 16             \\ \cline{2-3}
		&$B_1=B_2$                  &  1/6 or 1/12         \\ \cline{2-3}
            &$P_{\mathrm{total}_1}=P_{\mathrm{total}_2}$                  &  0.5         \\ \hline
		\multirow{3}{*}{Training parameters} & Epoch               & 200            \\ \cline{2-3} 
		& Initial learning rate       &     $10^{-4}$          \\ \cline{2-3} 
		& Batch size          &      8          \\ \hline
	\end{tabular}
\end{table}

\subsubsection{Simulation Details}
Simulation parameter settings are given in Table \ref{table0}. We compare the proposed RDJSCC with CAM-based DJSCC \cite{Wang2022ICASSP} under the same settings. The initial learning rate is $10^{-4}$, with a 50\% decay after 100 training epochs. Mixed SNRs training is employed, with SNR sampled from a uniform distribution [-8, 2] dB \footnote{{To emulate severe fading conditions, we focus on the low SNR regime ($\leq$ 2 dB) in our simulations.}}.

\begin{figure*}[!t]
	\centering
	\subfigure[]{
		\label{Experiment 1}
		\includegraphics[height=4cm, width=5cm] {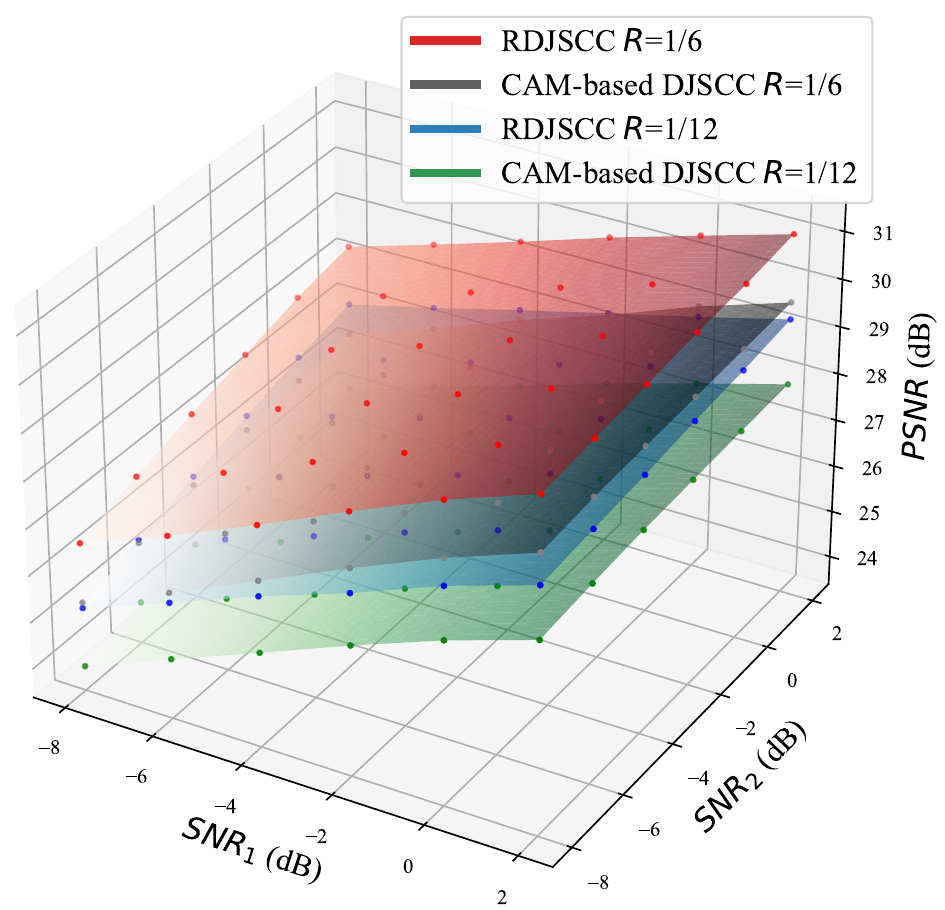}}
	\subfigure[]{
		\label{Experiment 2}
		\includegraphics[height=4cm, width=5cm] {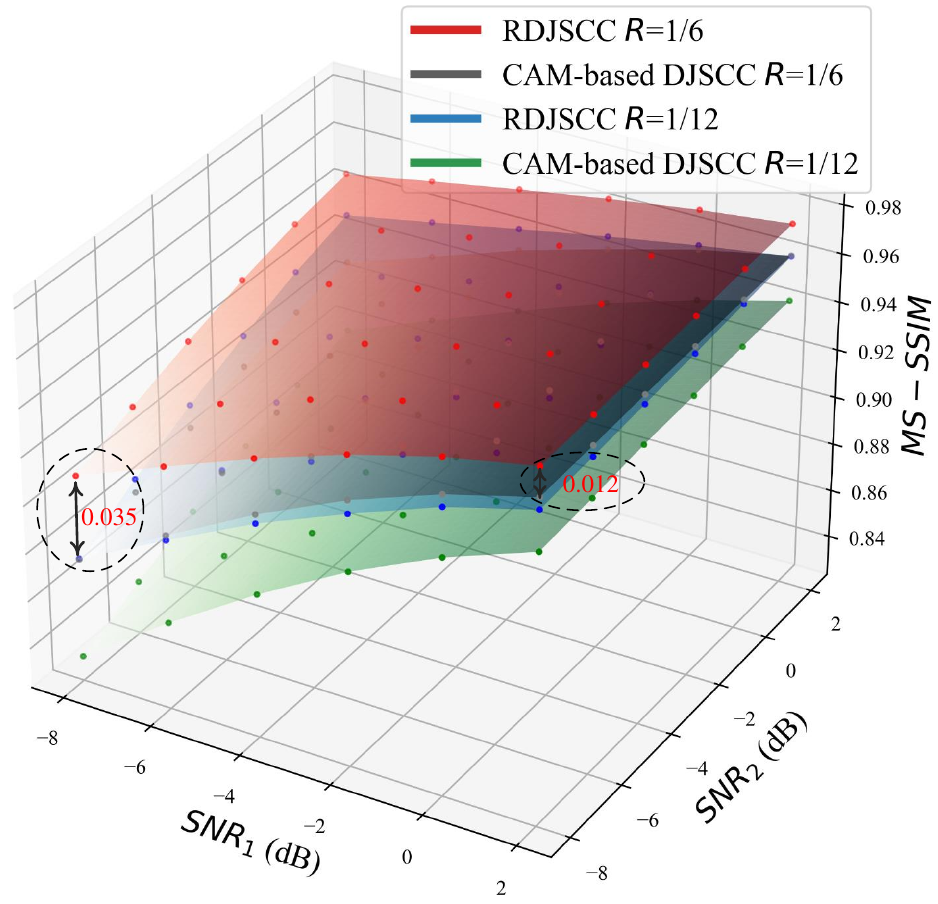}}
	\subfigure[]{
		\label{Experiment 3}
		\includegraphics[height=4cm, width=5cm] {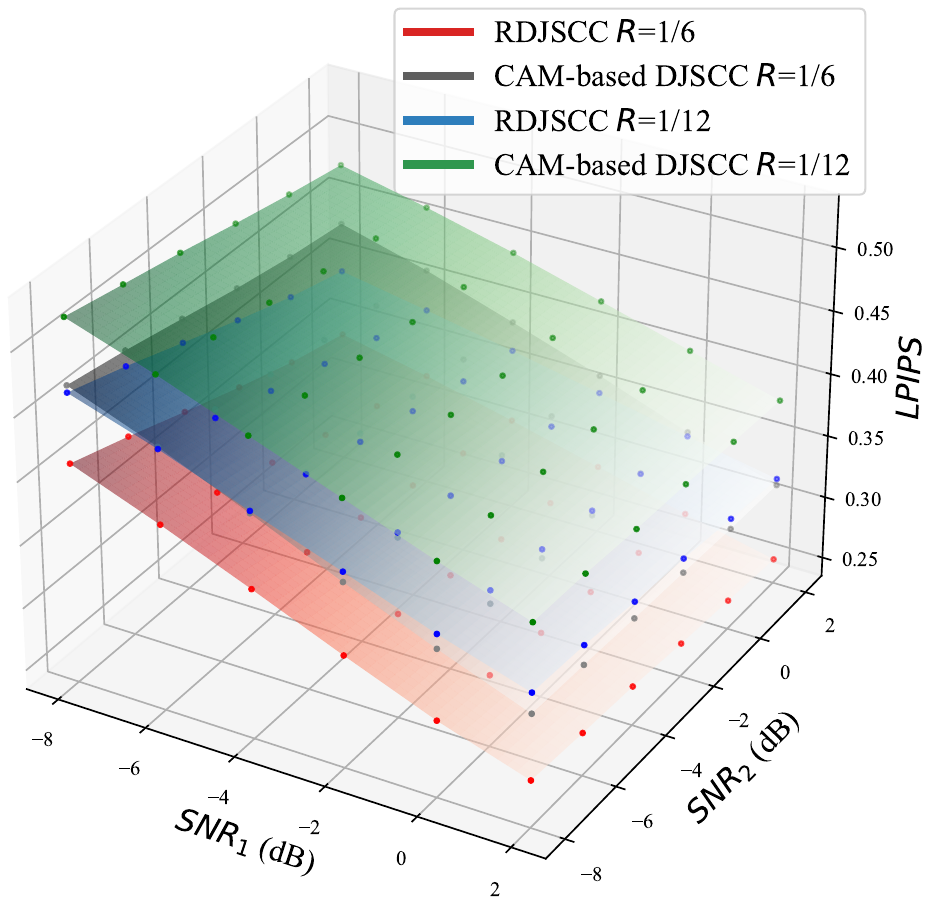}}
		\subfigure[]{
		\label{Experiment 4}
		\includegraphics[height=4cm, width=5cm] {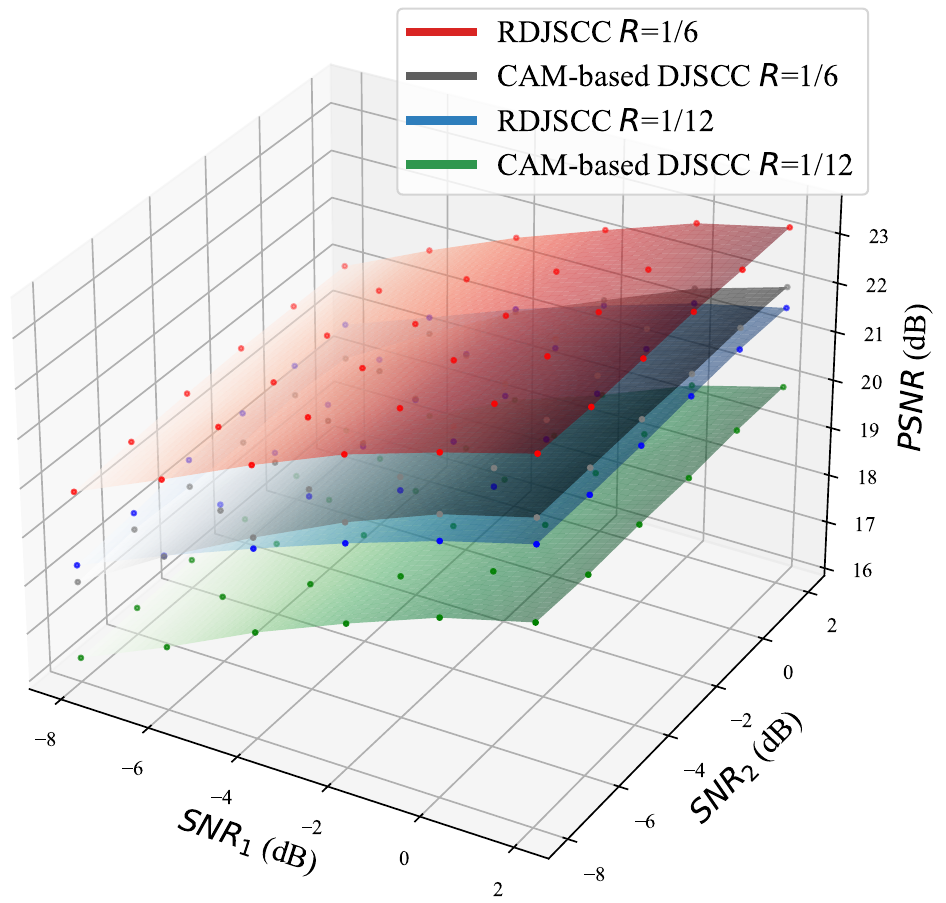}}
	\subfigure[]{
		\label{Experiment 5}
		\includegraphics[height=4cm, width=5cm] {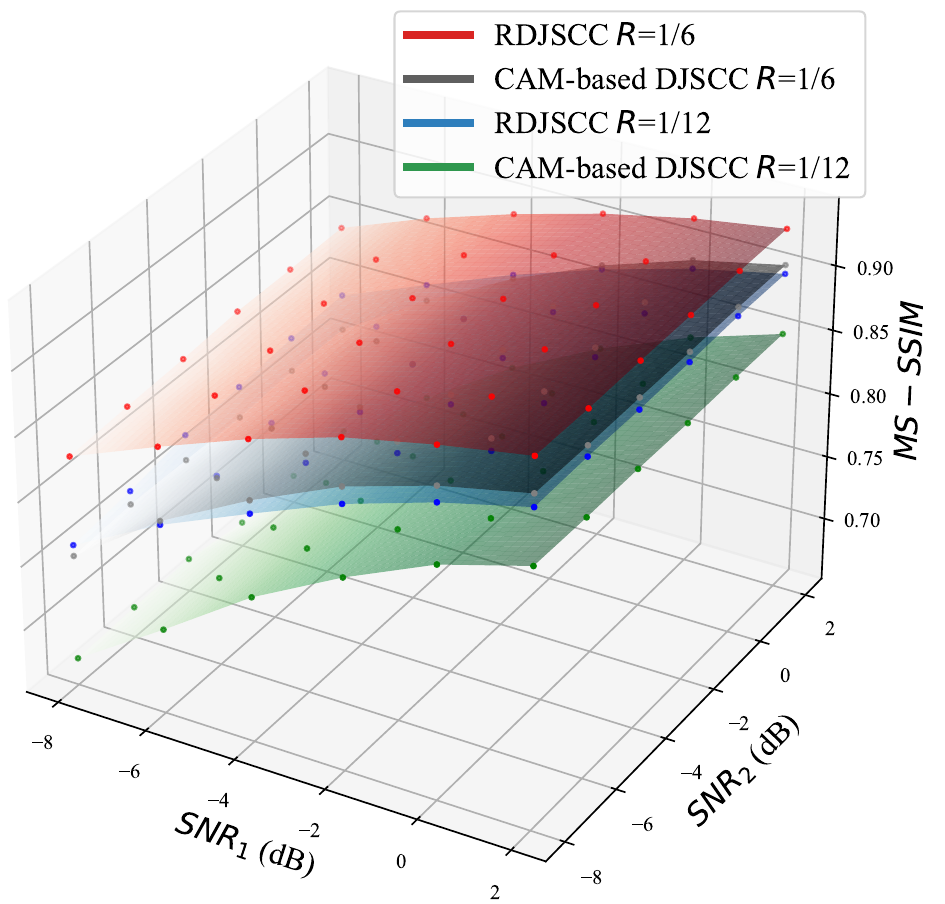}}
	\subfigure[]{
		\label{Experiment 6}
		\includegraphics[height=4cm, width=5cm] {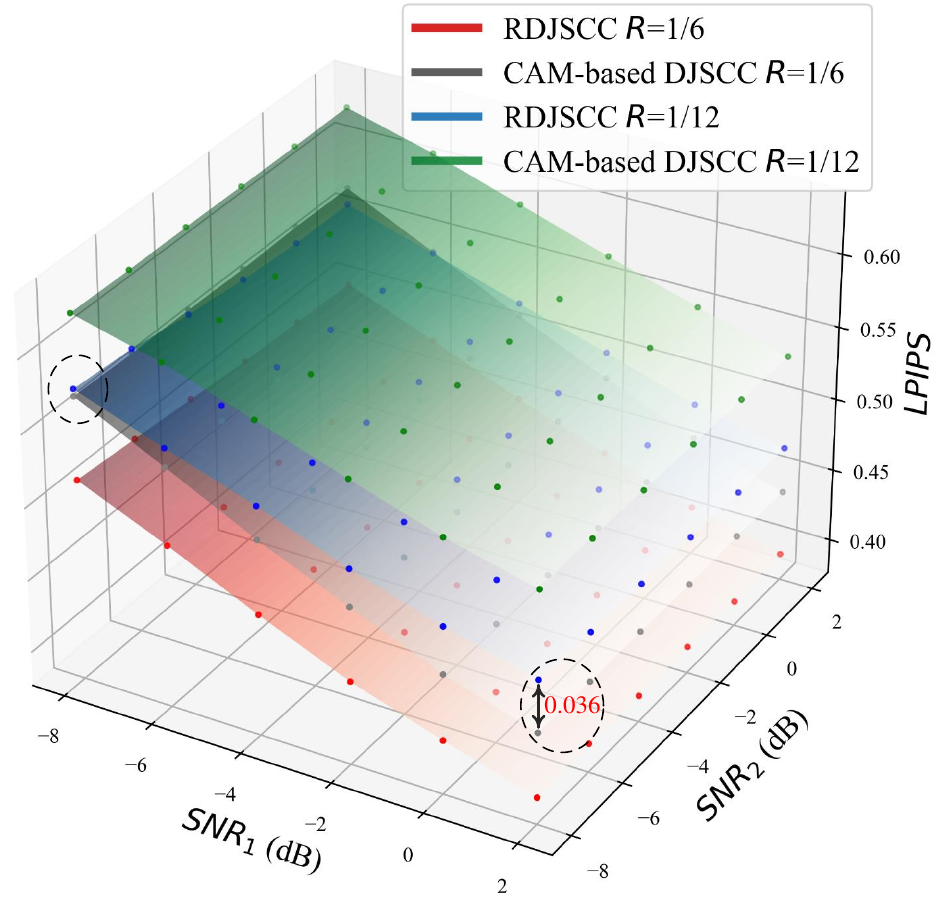}}
	\caption{Validating the effectiveness of the proposed RDJSCC: (a)-(c) Reconstruction performance of different methods when adopting PSNR, MS-SSIM and LPIPS as metrics respectively under Cityscapes dataset. (d)-(e) the Reconstruction performance under the KITTI dataset. 
	}
        \label{reconstruction_performance}
\end{figure*}

\subsection{Performance Metrics}\label{sec:Metrics}
\subsubsection{Distributed Image Transmission Performance}
We use PSNR, MS-SSIM and LPIPS to comprehensively evaluate the distributed image transmission performance of the proposed model. 
PSNR is a commonly used metric based on the pixel-wise MSE between input $\mathbf{s}$ and output $\mathbf{\hat s}$ as
\begin{equation}
	\label{PSNR}
	\mathrm{PSNR}(\mathbf{s},\mathbf{\hat s})\triangleq 10\log_{10}\frac{255^2}{d(\mathbf{s},\mathbf{\hat s})} \mathrm{(dB)}.
\end{equation}

MS-SSIM is based on the similarity of images at different scales, which is more in line with human visual perception \cite{Wang2004ACS}. The similarity is calculated by structural similarity (SSIM) as 
\begin{equation}
	\label{SSIM}
	\mathrm{SSIM}(\mathbf{s},\mathbf{\hat s})\triangleq\frac{(2\mu_{\mathbf{s}}\mu_{\mathbf{\hat s}}+C_{1})(2\sigma_{\mathbf{s}\mathbf{\hat s}}+C_{2})}{(\mu_{\mathbf{s}}^{2}+\mu_{\mathbf{\hat s}}^{2}+C_{1})(\sigma_{\mathbf{s}}^{2}+\sigma_{\mathbf{\hat s}}^{2}+C_{2})},
\end{equation}
where $\mu_{\mathbf{s}}$ and $\mu_{\mathbf{\hat s}}$ are the means of images. $\sigma_{\mathbf{s}}^2$ and $\sigma_{\mathbf{\hat s}}^2$ are the variances. $\sigma_{\mathbf{s}\mathbf{\hat s}}$ is the covariance. $C_1$ and $C_2$ are constants used to stabilize the formula. Based on SSIM, MS-SSIM can be defined as
\begin{equation}
	\label{MS-SSIM}
	\mathrm{MS}\textendash\mathrm{ SSIM}(x,y)\triangleq\left(\prod_{j=1}^M\mathrm{SSIM}_j(x,y)\right)^{\frac1M},
\end{equation}
where $M$ is the number of scales.

LPIPS computes the dissimilarity within the feature space between the input image and the reconstruction, by leveraging a deep neural network, we have 
\begin{equation}
	\label{LPIPS}
	{\mathrm{LPIPS}(\mathbf{s}, \mathbf{\hat s})\triangleq\sum_{i=1}^{I}\frac1{H_iW_i}\sum_{h,w}\|\boldsymbol{w}^{(i)}\odot(\boldsymbol{y}^{(i)}-\boldsymbol{\hat{y}}^{(i)})\|_2^2,}
\end{equation}
{where $\boldsymbol{y}^{(i)},\boldsymbol{\hat{y}}^{(i)}\in\mathbb{R}^{H_i\times W_i\times C_i}$ are the intermediate features derived from the $i$-th layer of the employed network with $I$ layers.} 
$H_i,W_i$, and $C_i$ are the intermediate height, width, and channel dimensions of features, respectively,  $\boldsymbol{w}^{(i)}\in\mathbb{R}^{C_i}$ is the weight vector, and $\odot$ is the
channel-wise feature multiplication operation.

\subsubsection{Algorithm Complexity}
We evaluate the complexity from two commonly used metrics: the number of floating-point operations (FLOPs) and parameters. In addition, we provide decoding time and model weight size as the auxiliary evaluation metrics.

\subsubsection{PAPR Performance}
We measure the PAPR performance of the considered communication system as the ratio between the maximum power and the average power
\begin{equation}
	\mathrm{PAPR}\{\mathscr{X}_{\mathrm {ofdm}}\}=\frac{\max\left|\mathscr{X}_{\mathrm {ofdm}}\right|^2}{E\left\{\left|\mathscr{X}_{\mathrm {ofdm}}\right|^2\right\}}.
\end{equation}

\subsection{Distributed Image Transmission Performance}
% To validate the effectiveness of our proposed framework, we assess it from three perspectives: performance analysis, PAPR reduction and model capacity analysis, which corresponds to the trade-offs discussed in Section \ref{sec:Metrics}.
\subsubsection{Performance Analysis}
 Fig. \ref{reconstruction_performance} shows the performance of different methods under different SNRs, compression ratios, and datasets (a higher PSNR/MS-SSIM or a lower LPIPS indicates a better performance). Since the two views are symmetric, we only report the performance of one view transmitted with $SNR_1$. Fig. \ref{reconstruction_performance} shows that the proposed RDJSCC achieves better performance under different SNRs, compression ratios, and datasets compared with CAM-based DJSCC \cite{Wang2022ICASSP}. It indicates that RDJSCC can learn a more optimal  maximum likelihood estimation solution than CAM-based DJSCC.  
 % Fig. \ref{reconstruction_performance} also shows that RDJSCC has a high performance improvement in the terms of LPIPS. 
  Fig. \ref{Experiment 2} shows that RDJSCC has a greater performance gain at low SNR compared with CAM-based DJSCC. This is because RDJSCC can leverage correlated sources effectively and thus enhance the reconstruction performance under poor channel conditions.
 Fig. \ref{Experiment 6} shows that the LPIPS score of RDJSCC $R$=1/12 almost achieves or even exceeds the CAM-based DJSCC $R$=1/6. It illustrates that RDJSCC has a better human perception quality in terms of LPIPS.

\subsubsection{Model Complexity Analysis}
Table \ref{table2} shows the algorithm complexity. Results indicate that the complexity of RDJSCC is higher than that of the CAM-based DJSCC. This is because RDJSCC introduces additional computation overhead with a light MLP and the shift mechanism in CVIE and CCF. {Although the complexity of RDJSCC is not advantageous, the actual inference time may not vary proportionally. First, CAM-based DJSCC introduces cross-attention mechanism which is typically more time-consuming than convolution for image processing tasks. Secondly, based on the well-optimized CNNs or MLPs inference library (e.g., cuDNN), the decoding time of RDJSCC is close to the CAM-based DJSCC. }

\begin{table}[htb]
    \caption{The complexity comparison when $R$ = 1/6 and Cityscapes dataset is adopted}
    \centering
    \resizebox{\linewidth}{!}{
        \begin{tabular}{|cc|cccc|}
            % \toprule
            \hline
            & {\bf Methods} & {\bf FLOPs} & {\bf Parameters} & {\bf Weight} & {\bf Time} \\ \hline
            % \midrule
            & CAM-based DJSCC\cite{Wang2022ICASSP} & {1.66 G} & {3.22 M} & {14.58 MB} & {47.68 ms} \\
            & Proposed RDJSCC & {2.38 G} & {3.76 M} & {16.42 MB} & {46.36 ms} \\ \hline
            % \toprule
        \end{tabular}
    }
    \label{table2}
\end{table}

For a further analysis, theoretical complexity analysis is given. RDJSCC and CAM-based DJSCC both adopt individual $1 \times 1$ convolution to generate query, key, and value. The parameters of $1 \times 1$ convolution are $3C^2$, where $C$ denotes the number of convolutional channels. The computational complexity is $\mathcal{O}(3HWC^2)$, where $H$ and $W$ denote the height and width of the input features. 
% Apart from that, CAM has no other parameters (the calculation of Eq .(\ref{Cross-attention-mechanism}) does not require parameters but FLOPs). 
The total computational complexity of CAM is $\mathcal{O}(3HWC^2+2HWK^2C)$. Compared with CAM, the additional parameters of RDJSCC brought by MLP are $3K^2$. It is much less than $3C^2$ (e.g., with $C$ often set empirically as 256 and $K$ as 3 or 5). The additional parameters of RDJSCC brought by the shift mechanism are $K^4C$, which is also light compared with $1 \times 1$ convolutions. Meanwhile, $1 \times 1$ convolutions are reused in CCF mechanism, which reduces computational overhead. The total computational complexity of RDJSCC is $\mathcal{O}(3HWC^2+3HWK^2C+K^4C)$. 
% The above analysis indicates that the design of RDJSCC is compact. 

\begin{figure*}[!t]
	\centering
	\subfigure[]{
		\label{Experiment 7}
		\includegraphics[width=0.22\textwidth] {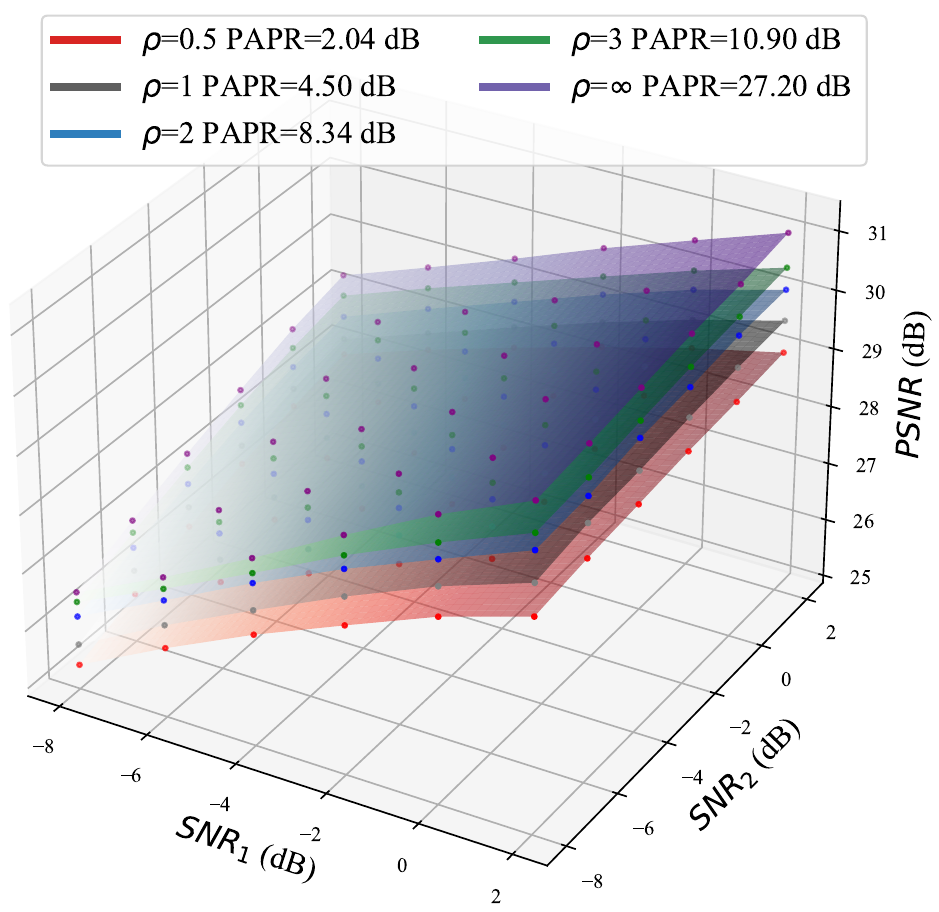}}
        \subfigure[]{
		\label{Experiment 8_1}
		\includegraphics[width=0.22\textwidth] {{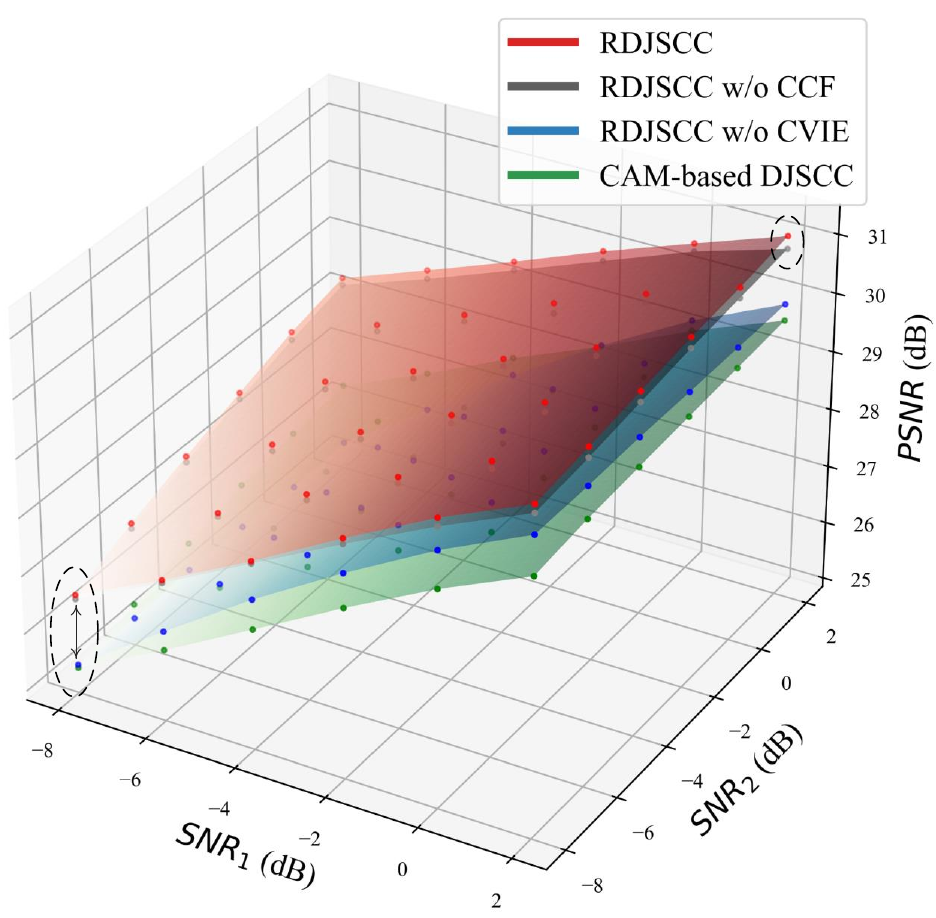}}}
        \subfigure[]{
		\label{Experiment 8_2}
		\includegraphics[width=0.23\textwidth] {{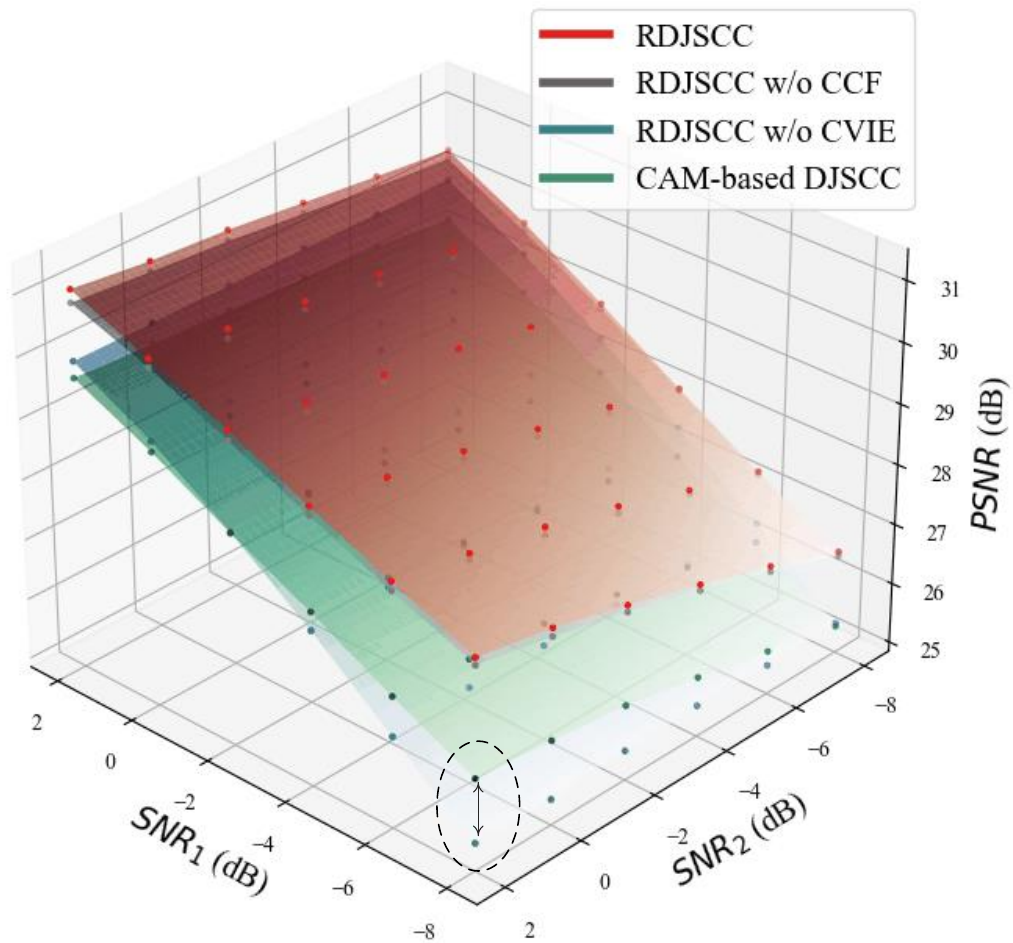}}}
	\subfigure[]{
		\label{Experiment 9}
		\includegraphics[width=0.25\textwidth] {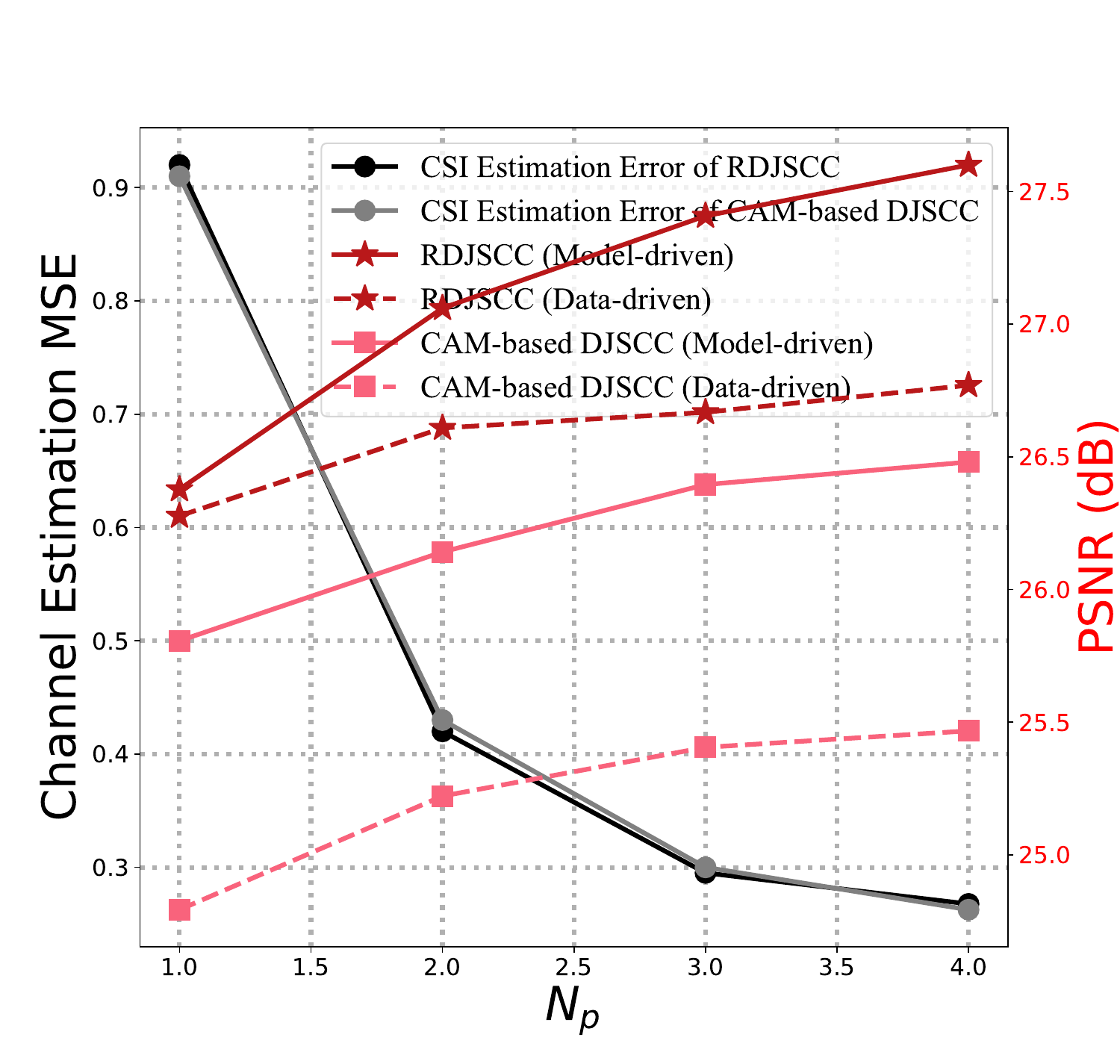}}
	\caption{(a) PAPR under different clipping ratios. {(b)-(c) Ablation study results in terms of the PSNR performance. To better present the 3D result, we separately show the two sides of the same image as illustrated in (b) and (c). (d) PSNR Performance with respect to the number of pilot symbols. }}
        \label{ex2}
\end{figure*}

\subsubsection{PAPR Reduction}
Fig. \ref{Experiment 7} shows the trade-off between PAPR and performance by clipping. It is obvious that the reconstruction performance decreases as clipping ratio decreases. This is because clipping disrupts the orthogonality among subchannels. However, we can strike a trade-off between performance and PAPR by setting an appropriate clipping ratio. For example, the reconstruction performance with $\rho = 3.0$ is nearly identical to the performance without clipping ($\rho = \infty$). The results are consistent with \cite{Yang2022TCCN,Shao2023WCL}. It indicates that the performance gain brought by OFDM against fading and the low PAPR can coexist in a distributed image transmission system.

\subsubsection{Ablation Study}
During the design process of RDJSCC, we use the CVIE to replace the CAM. Meanwhile, we have designed the CCF based on CVIE. To verify the effectiveness of these methods, we report the transmission performance in the following settings: we take off the CVIE and CCF respectively (marked as ``RDJSCC w/o CVIE'' and ``RDJSCC w/o CCF'' in Fig. \ref{Experiment 8_1}). Note that once CVIE is removed, no correlated sources are accessed. 

Fig. \ref{Experiment 8_1} shows that RDJSCC w/o CVIE has a performance degradation compared with RDJSCC. It indicates that CVIE can utilize correlated sources and thus enhance transmission performance. 
{Fig. \ref{Experiment 8_1} also shows that RDJSCC has a performance gain compared with RDJSCC w/o CCF. It demonstrates the benefit of fusing the complementarity and consistency through CCF. It also verifies that the performance gain of RDJSCC arises from the combination of CVIE and CCF. Fig. \ref{Experiment 8_2} shows an interesting trend that the CAM performs better when the channel condition of the other view is relatively good (i.e. when $SNR_2$ = 2 dB). In contrast, Fig. \ref{Experiment 8_1} shows that the CAM performs worse when the channel condition of the other view is poor (i.e. when $SNR_2$ = -8 dB). This indicates that CAM struggles to efficiently utilize noisy correlated sources, as Remark~\ref{remark5}.}

\subsubsection{Evaluation of CSI Estimation Errors}
{As shown in Fig. \ref{Experiment 9}, we evaluate the robustness of RDJSCC under varying levels of CSI estimation error and compare it with the CAM-based DJSCC benchmark. Two CSI estimation strategies are considered: model-driven and data-driven. The model-driven approach follows \cite{Yang2022TCCN}, where CSI is explicitly estimated via MMSE and optimized jointly with the decoder. In contrast, the data-driven method implicitly estimates CSI by directly feeding pilot symbols and received signals into the decoder without explicit supervision. To quantify the CSI estimation error, we vary the number of pilot symbols $N_p$. As expected, the model-driven approach outperforms the data-driven one, consistent with prior observations in \cite{Yang2022TCCN}. More importantly, under different CSI estimation error levels, RDJSCC consistently outperforms CAM-based DJSCC in both settings, demonstrating its superior adaptability to CSI inaccuracies.}

\subsection{Visual Comparison}
In Fig. \ref{reconstruction_performance}, we have already shown the reconstruction performance of the proposed method under different datasets and compression ratios. To visually illustrate the impact of fading channel, Fig. \ref{Visual}(a) presents examples of the reconstructed images. Results show that the proposed method {presents} a better recovery quality compared with CAM-based DJSCC.

Fig. \ref{Visual}(a) also shows that the semantic features of the encoder outputs from both methods differ significantly, despite the identical encoder structure. This difference occurs because the complexity of the decoder affects the encoder's learning process. In other words, a more complex decoder might necessitate richer and more intricate feature maps, whereas a simpler decoder might require only basic feature maps. Fig. \ref{Visual}(a) shows that the shallow features (e.g. the output of $E_1$) retain most of the details of the original image. As the encoder depth increases, the feature maps become more random and abstract, reflecting the optimization process of joint source-channel coding. 

 \begin{figure}[htbp]
	\centering
	\includegraphics[width=0.45\textwidth] {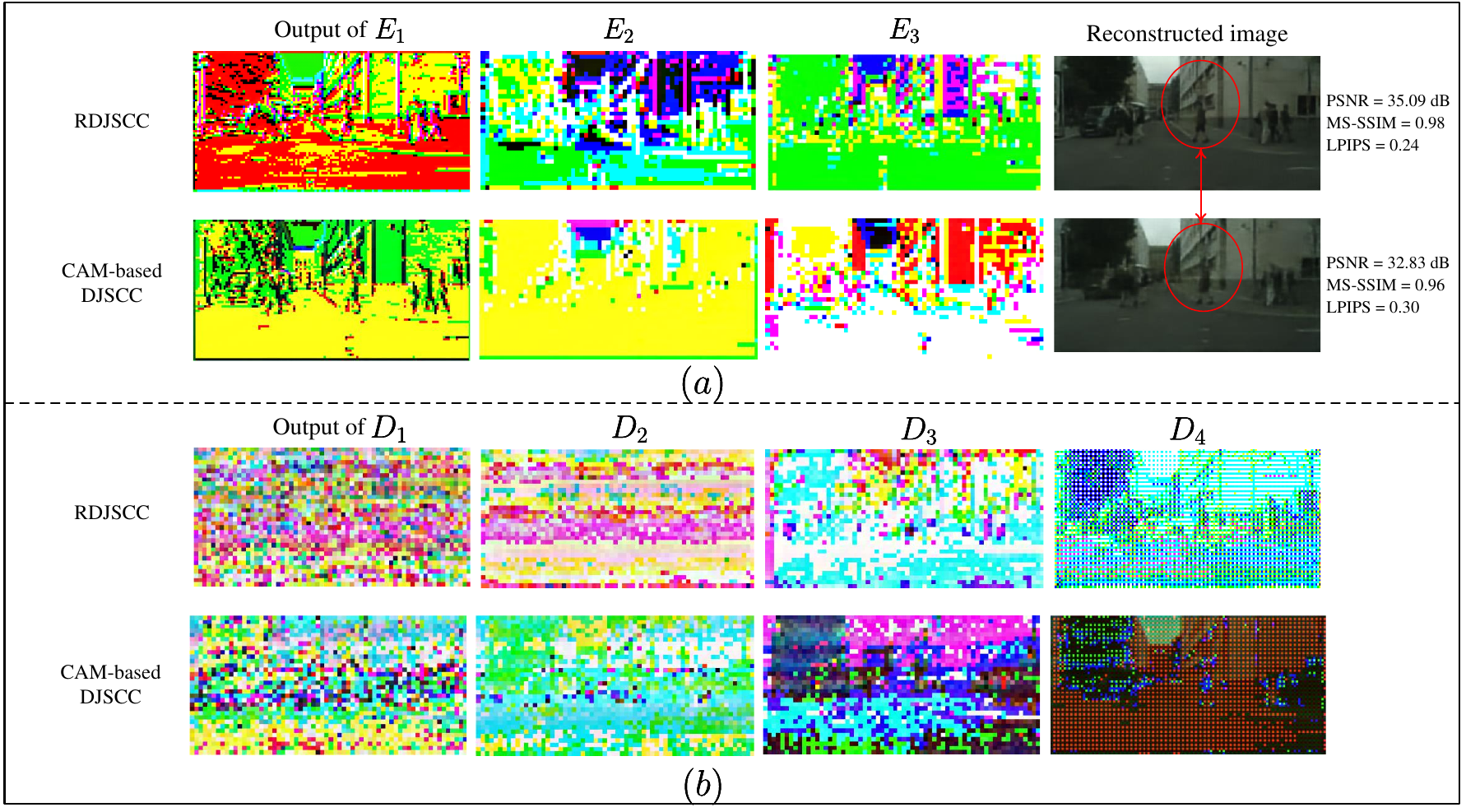}
	\caption{(a) The output semantic features of different stacked blocks at the encoder and the reconstructed image. (b) The output semantic features of different stacked blocks at the decoder.}
	\label{Visual}
\end{figure}

In Fig. \ref{Visual}(a), the deep encoded features (e.g. the output of $E_3$) also exhibit a visible structure. This indicates that a dependency structure exists in the semantic feature map, consistent with \cite{Liang2024TWC}. According to \cite{Balle2018arXiv}, modeling these dependencies by introducing latent variable can enhance lossy image compression performance in noise-free environments. However, such dependencies may also render signals more resilient to interference when transmitting in noisy channels, as shown in Fig. 9 of \cite{Liang2024TWC}. Hence, dependencies can be regarded as a component of reconstruction-relevant information. Fig. \ref{Visual}(b) indicates that the decoding process exhibits the inverse trend compared with encoding process, transitioning from random to detailed textures.
% In addition, RDJSCC can better recover details compared to CAM-based DJSCC by observing the outputs of $D_4$ for both.

%  \begin{figure}[htbp]
\begin{figure}[htbp]
	\centering
	\subfigure[]{
		\label{Weight 1}
		\includegraphics[height=3.0cm, width=3.5cm] {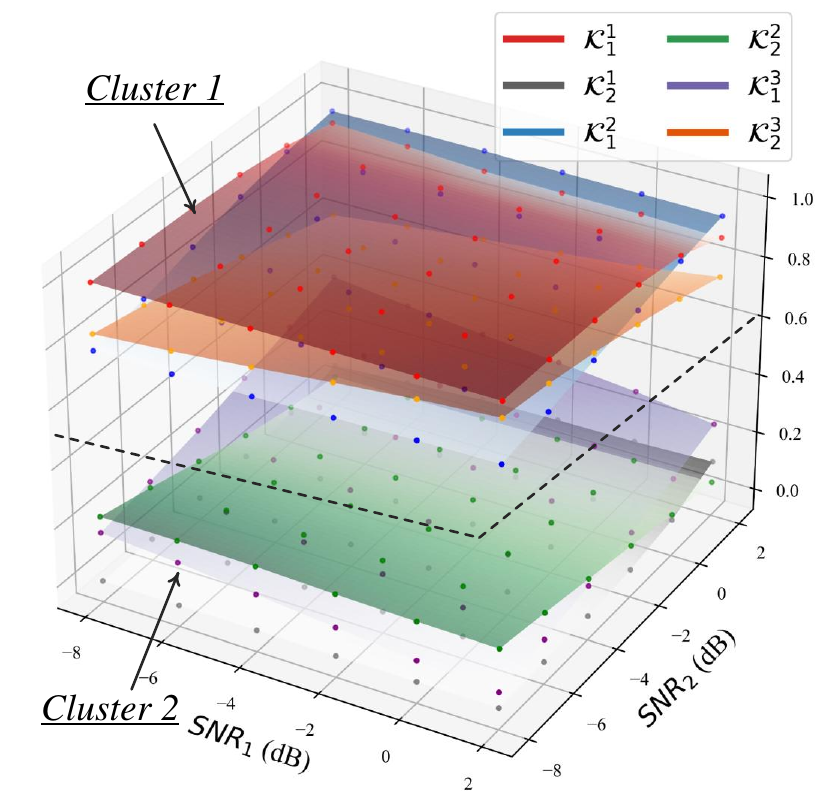}}
	\subfigure[]{
		\label{Weight 2}
		\includegraphics[height=3.0cm, width=3.2cm] {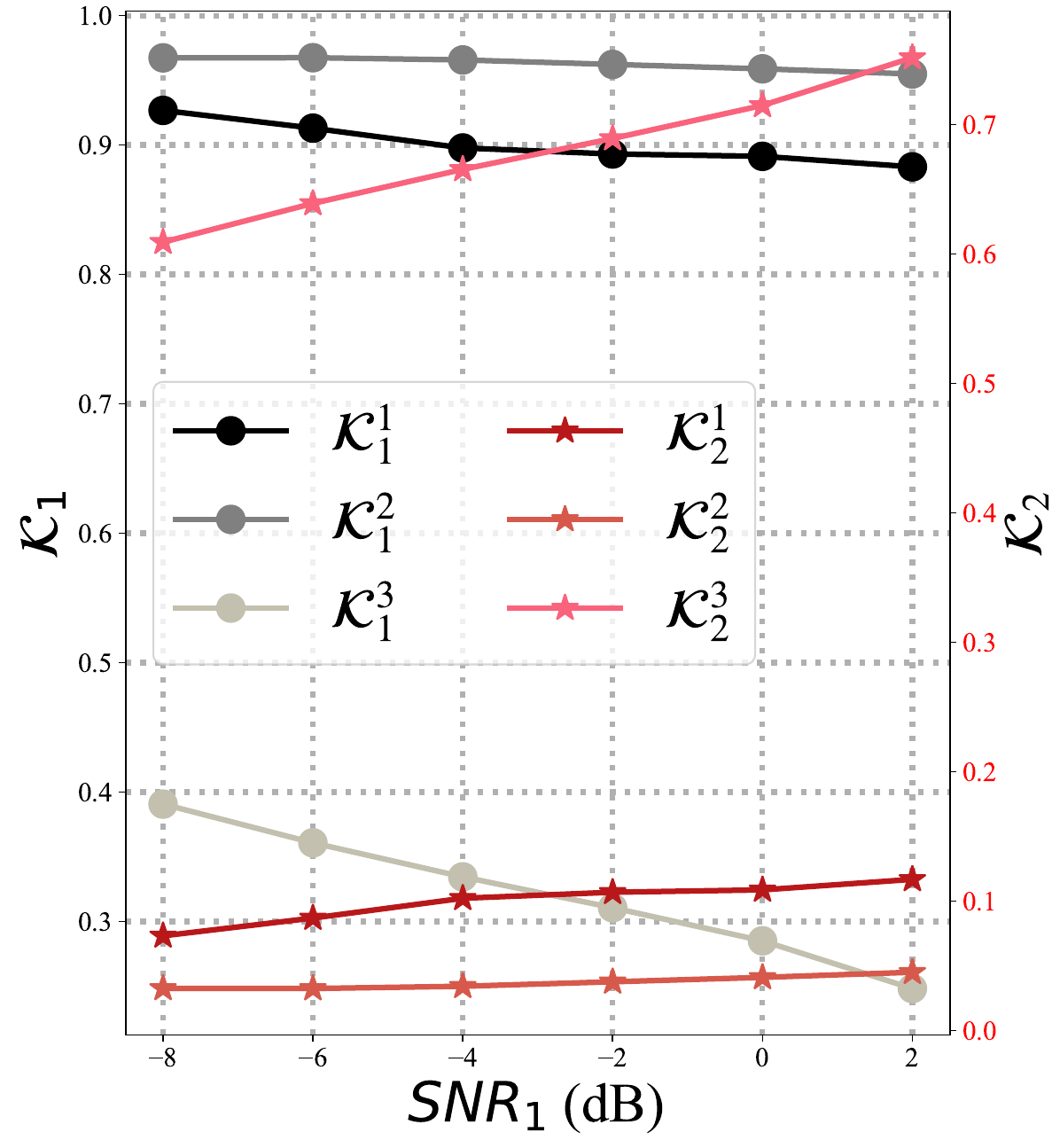}}
	\caption{(a) Weight allocation of $\mathcal{K}_1$ and $\mathcal{K}_2$ under different SNRs when $R$ = 1/6 and Cityscapes dataset is aopted. $\mathcal{K}_1^i$ and $\mathcal{K}_2^i$ respectively denote the weight assigned to complementarity and consistency of the $i$-$th$ stacked blocks $D_i$. {(b) Weight allocation of $\mathcal{K}_1$ and $\mathcal{K}_2$ when $SNR_2$ is fixed at 2 dB.}}
	\label{Weight}
\end{figure}

\subsection{Discussions on Dynamic Weight}
Fig \ref{Weight 1} gives the weights $\mathcal{K}_1$ and $\mathcal{K}_2$ under different SNRs. These weights are clearly divided into two clusters. It indicates that one weight tends to dominate. Specifically, Fig \ref{Weight 2} gives an interesting case where $SNR_2$ is fixed at 2 dB. It reveals that there are more complementary weight assigned in the shallow decoding blocks, i.e., $\mathcal{K}_1^1 \textgreater \mathcal{K}_2^1$. However, the opposite result appears in the deep layer decoding blocks, i.e., $\mathcal{K}_1^3 \textless \mathcal{K}_2^3$. In other words, the shallow decoding blocks focus on complementarity, while the deep decoding blocks focus on consistency. As shown in Fig \ref{Visual}(b), the shallow features are more random, with visible structure emerging only in the third decoding block $D_3$, which coincides with an increased emphasis on consistency. This phenomenon indicates that the shallow decoding blocks primarily act as denoisers, with consistency feature extraction {occurs} once the noise is sufficiently reduced. Furthermore, we observe that the deep decoding blocks allocate more resources to consistency as the SNR increases, suggesting that in lower noise conditions, utilizing correlated sources enhances performance.

\section{Conclusion} \label{sec:Conclusion}
In this paper, we propose a novel RDJSCC scheme, specifically designed for distributed image transmission under multi-path fading channel without perfect CSI. We aim to efficiently utilize noisy correlated sources to achieve collaborative recovery, striking a balance between complementarity and consistency. Based on the information-theoretic analysis, we find that designing flexible multi-view transmission framework to meet the requirements of consistency and complementarity can lead to the optimal reconstruction performance. Deploying CVIE and CCF at the decoder can strike a graceful trade-off between reconstruction complementarity and consistency. Meanwhile, we also verify that the low PAPR can coexist with performance in distributed image transmission system by setting proper clipping threshold. We believe that the proposed RDJSCC scheme will play a pivotal role in advancing distributed systems towards practical deployment.
\appendices

\section{DNN Architecture and Hyper-parameters} \label{sec:Architecture}
The DNN architectures are given in Fig. \ref{TWC_proposed}. The encoder and decoder are both composed of stacked residual blocks as \cite{Zhang2023TWC} shown in Fig. \ref{TWC_proposed}(b). Specifically, the encoder consists of three residual blocks. The latter two residual blocks performing downsampling twice. The encoders for the two views share the same parameters and the encoding process are independent, which refers to Slepian-Wolf theorem on distributed source coding \cite{Wolf1973TIT}. It proves that separate encoding and joint decoding of two or more correlated sources can theoretically achieve the same compression ratio as a joint encoding-decoding scheme under lossless compression, which has been extended to the lossy transmission scenario \cite{WynerZiv1976, Heegard1985}. 
% Based on above information theory, The two views use independent encoders with shared parameters and a single decoder for collaborative reconstruction. 
The encoder is simpler compared with the decoder, because DSC enables low-complexity encoding by shifting a significant amount of computation to the decoder. Fig. \ref{TWC_proposed} shows that the decoder consists of five residual blocks while the encoder consists of three.
\begin{figure}[htbp]
	\centering
	\includegraphics[height=3cm, width=8cm] {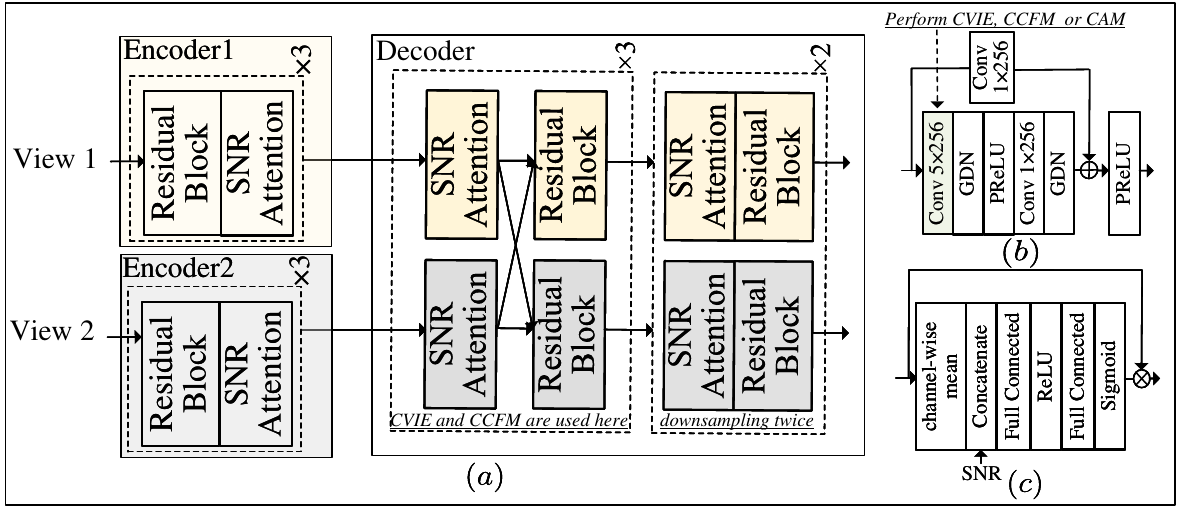}
	\caption{(a) ResNet-based encoder and decoder. (b) Stacked residual blocks. (c) SNR adaptive strategy. }
	\label{TWC_proposed}
\end{figure}

Meanwhile, we also introduce an SNR adaptive strategy to use a single $(\boldsymbol{\phi},\boldsymbol{\theta})$ pair for various SNRs shown in Fig. \ref{TWC_proposed}(c), which has been widely proved to be beneficial for channels with varying SNR\cite{Zhang2023TWC,Xu2022TCSVT,Wang2022ICASSP}. Specifically, channel-wise mean and SNR are concatenated. As shown in Fig. \ref{TWC_proposed}(a), CVIE and CCF are stacked before upsampling to avoid calculating the cross-view 
dependence on high-dimensional features. 

\section{The Impacts of Encoding Process on Correlation} \label{sec:Impacts}
As depicted in Section \ref{sec:Fusion}, analytically calculating the MI of multi-view encoded representations is non-trivial due to the high nonlinearity of the DNN. We also find that the MI of received encoded representations between two views can be roughly predicted according to the SCS. Here, we attempt to analyse the impacts of encoding process on correlation based on cosine similarity. We use a canonical correlation analysis (CCA) \cite{Hardoon2004TWC} to have a further insight of the discussions in Section \ref{sec:Correlation}. 

By definition, the cosine similarity between two views is given by
\begin{equation}
	\label{SCS_eq}
	\mathrm{cos}(\psi) = \frac{\langle \mathbf{s}_1^j, \mathbf{s}_2^j \rangle}{\|\mathbf{s}_1^j\| \|\mathbf{s}_2^j\|},
\end{equation}
where $\psi$ denotes the angle between two vectors. Here, the correlated sources from two views are treated as vectors in a common feature space. Since the encoding process $f(\cdot; \boldsymbol{\phi})$ is nonlinear (composed of multiple sub-nonlinear transformation $E_j$), we resort to kernel CCA analysis. Kernel CCA offers an analytical framework by projecting the data into a higher dimensional feature space with a fixed nonlinear mapping as $\Phi:\mathbf{s}\mapsto\Phi(\mathbf{s})$ \cite{Hardoon2004TWC},  
{where $\Phi$ is a mapping from the input space $\mathbf{S}$ to feature space $\mathbf{F}$.}

\begin{assumption}
	\label{assumption3}
    $E_j$ can be integrated into the kernel function or have a compatible kernel representation, thus we can redefine the kernel as $K_j(\mathbf{s}_1,\mathbf{s}_2)=E_j(\Phi(\mathbf{s}_1))\cdot E_j(\Phi(\mathbf{s}_2))$.
\end{assumption}

Under the Assumption \ref{assumption3}, the cosine similarity of two views can be expressed as
\begin{equation}
    \begin{aligned} 
    \mathrm{cos}(\psi_1)=\frac{K_j(\mathbf{s}_1,\mathbf{s}_2)}{\sqrt{K_j(\mathbf{s}_1,\mathbf{s}_1)}\sqrt{K_j(\mathbf{s}_2,\mathbf{s}_2)}},
    \end{aligned} 
\end{equation}
where $\psi_1$ is the angle after transformation. By the Lagrange multiplier method, we can maximize the cosine similarity and obtain a generalised eigenproblem with the form $A\mathbf{x} = \lambda\mathbf{x}$\cite{Hardoon2004TWC}. Under this case, $E_j=\frac1\lambda {K_j(\mathbf{s}_2,\mathbf{s}_2)}^{-1}{K_j(\mathbf{s}_1,\mathbf{s}_1)}$ is optimal. As a
conclusion, the changes in cosine similarity are related to the selection of kernel function. Perfect correlation can be formed when kernel function is invertible, 

It is worth mentioning that when $E_j$ is a linear transformation, the cosine similarity can be expressed as 
% (singular value decomposition (SVD) on $E_j$)
\begin{equation}
    \begin{aligned} 
    \mathrm{cos}(\psi_1)=\frac{({\mathbf{s}_1^j})^{T}\Sigma ({\mathbf{s}_2^j})}{\sqrt{({\mathbf{s}_1^j})^{T}\Sigma ({\mathbf{s}_1^j})({\mathbf{s}_2^j})^{T}\Sigma ({\mathbf{s}_2^j})}},
    \end{aligned} 
\end{equation}
where $\Sigma={E_j}^T{E_j}$. The optimal cosine similarity can be computed by the singular value decomposition \cite{whuber2016}.

\begin{IEEEbiography}[{\includegraphics[width=1in,height=1.25in, clip,keepaspectratio]{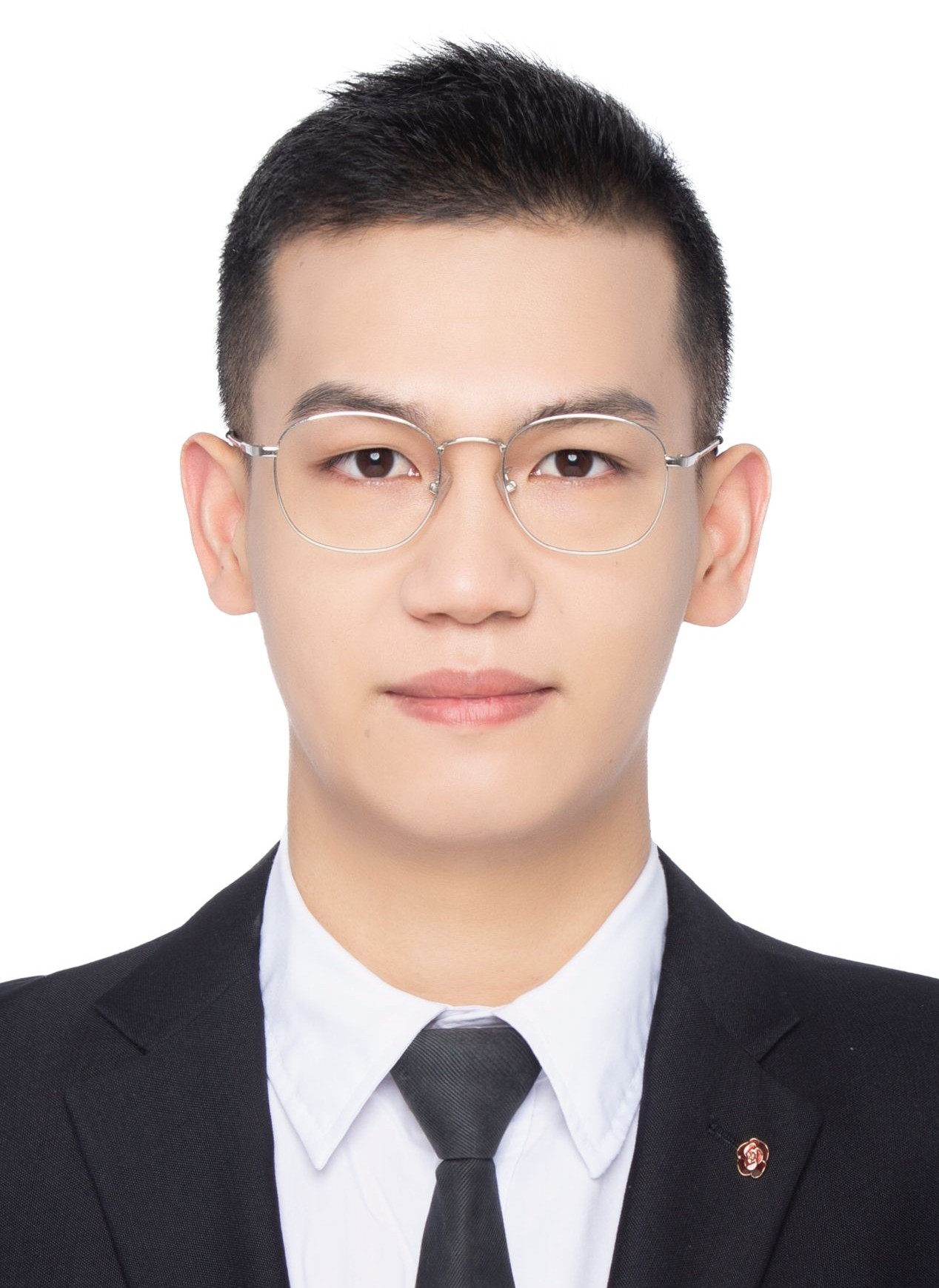}}]{Biao Dong} received the M.Eng. degree from Nanjing University of Posts and Telecommunications, Nanjing, China, in 2020. He is currently pursuing the Ph.D. degree in Harbin Institute of Technology (Shenzhen), Shenzhen, China. 

His research interests include wireless communications, signal processing and machine learning.
\end{IEEEbiography}

\begin{IEEEbiography}[{\includegraphics[width=1in,height=1.25in, clip,keepaspectratio]{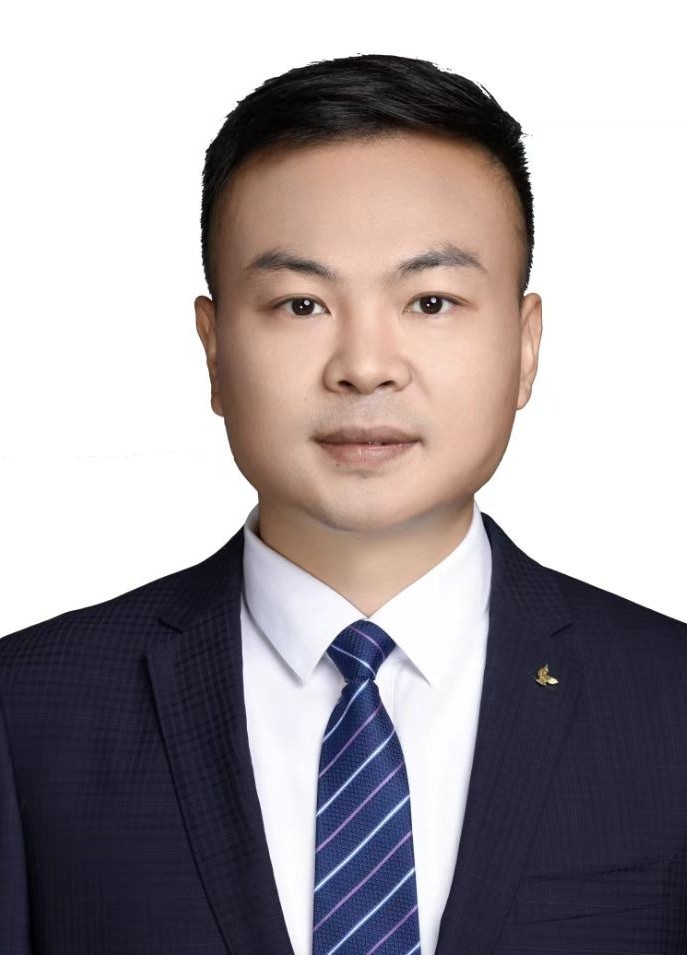}}]{Bin Cao} (Member, IEEE) received the Ph.D. degree in information and communication engineering from Harbin Institute of Technology, Shenzhen, China, in 2013. He is currently an Associate Professor with the School of Electronics and Information Engineering, Harbin Institute of Technology. From 2010 to 2012, he was a Visiting Scholar with the University of Waterloo, Waterloo, ON, Canada. 

His research interests include signal processing for wireless
communications, cognitive radio networking, and resource allocation for wireless networks.
\end{IEEEbiography}

\begin{IEEEbiography}
[{\includegraphics[width=1in,height=1.25in,clip,keepaspectratio]{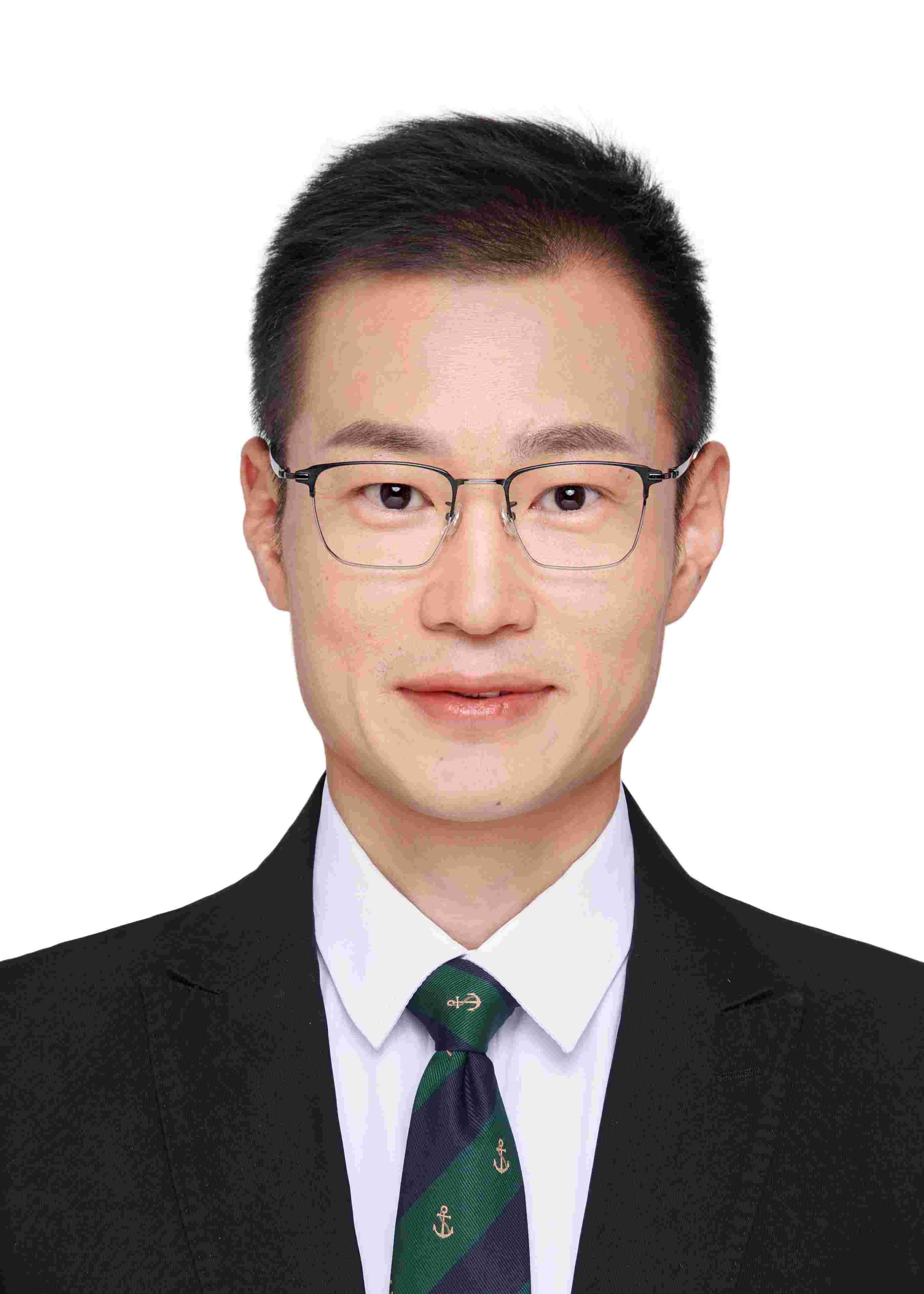}}]
{Guan Gui} (Fellow, IEEE) received the Ph.D. degree from the University of Electronic Science and Technology of China, Chengdu, China, in 2012. From 2009 to 2014, he joined Tohoku University as a research assistant and postdoctoral research fellow. From 2014 to 2015, he was an Assistant Professor at the Akita Prefectural University, Akita, Japan. Since 2015, he has been a professor at Nanjing University of Posts and Telecommunications, Nanjing, China. His recent research interests include intelligence sensing and recognition, intelligent signal processing, and physical layer security. Dr. Gui has published over 200 IEEE Journal/Conference papers and won several best paper awards, e.g., ICC 2017, ICC 2014 and VTC 2014-Spring. He received the IEEE Communications Society Heinrich Hertz Award in 2021, the Clarivate Analytics Highly Cited Researcher in Cross-Field in 2021-2024, the Member and Global Activities Contributions Award in 2018, the Top Editor Award of IEEE Transactions on Vehicular Technology in 2019. Since 2022, he has been a Distinguished Lecturer of the IEEE Vehicular Technology Society. He serves or serves on the editorial boards of several journals, such as IEEE Transactions on Information Forensics and Security, IEEE Internet of Things Journal, and IEEE Transactions on Vehicular Technology. In addition, he served as the IEEE VTS Ad Hoc Committee Member in AI Wireless, Executive Chair of IEEE ICCT 2023, Executive Chair of VTC 2021-Fall, and Vice Chair of WCNC 2021.
\end{IEEEbiography}

\begin{IEEEbiography}[{\includegraphics[width=1in,height=1.25in, clip,keepaspectratio]{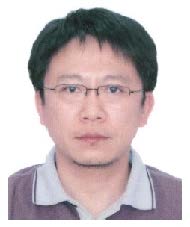}}]{Qinyu Zhang} (Senior Member, IEEE) received the bachelor’s degree in communication engineering from Harbin Institute of Technology (HIT), Harbin, China, in 1994, and the Ph.D. degree in biomedical and electrical engineering from the University of Tokushima, Tokushima, Japan, in 2003. From 1999 to 2003, he was an Assistant Professor with the University of Tokushima. From 2003 to 2005, he was an Associate Professor with Shenzhen Graduate School, HIT. He was the Founding Director of the Communication Engineering Research Center, School of Electronic and Information Engineering (EIE). Since 2005, he has been a Full Professor and the Dean of the EIE School, HIT. His research interests include aerospace communications and networks, wireless communications and networks, cognitive radios, signal processing, and biomedical engineering. He was the Associate Chair of Finance of the International Conference on Materials and Manufacturing Technologies in 2012. He was the TPC Co-Chair of the IEEE/CIC ICCC 2015. He was the Symposium Co-Chair of the CHINACOM 2011 and the IEEE Vehicular Technology Conference 2016 (VTC Spring). He was the Founding Chair of the IEEE Communications Society Shenzhen Chapter. He is on the editorial board of some academic journals, such as Journal of Communication, KSII Transactions on Internet and Information Systems, and Science China Information Sciences. He has been a TPC Member of the Infocom, IEEE ICC, IEEE GLOBECOM, IEEE Wireless Communications and Networking Conference, and other flagship conferences in communications.
\end{IEEEbiography}

\end{document}